\documentclass[11pt]{article}

\usepackage{times,graphicx,comment,listings,url,epstopdf,amsmath,algorithmicx,algorithm} 
\usepackage{color}
\usepackage{graphics}
\usepackage{graphicx}
\usepackage{subfigure}
\usepackage{xspace}
\usepackage{float}
\usepackage[noend]{algpseudocode}
\usepackage[compact]{titlesec}

\newtheorem{theorem}{Theorem}[section]
\newtheorem{lemma}[theorem]{Lemma}

\newcommand{\sq}{\hbox{\rlap{$\sqcap$}$\sqcup$}}
\newcommand{\qed}{\hspace*{\fill}\sq}
\newcommand{\set}[1]{\left\{#1\right\}}
\newcommand{\defn}[1]{\emph{\textbf{#1}}}
\newcommand{\card}[1]{\left|#1\right|}
\newenvironment{proof}{\noindent {\bf Proof.}\ }{\qed\par\vskip 4mm\par}

\begin{document}

\begin{titlepage}

\title{
Greedy Sequential Maximal Independent Set and Matching\\
are Parallel on Average
}
\author{Guy E. Blelloch\\
Carnegie Mellon University\\
guyb@cs.cmu.edu\\
\and
Jeremy T. Fineman\\
Georgetown University\\
jfineman@cs.georgetown.edu \\
\and 
Julian Shun\\
Carnegie Mellon University\\
jshun@cs.cmu.edu
}

\date{}

\maketitle \thispagestyle{empty}


\begin{abstract}
  The greedy sequential algorithm for maximal independent set (MIS)
  loops over the vertices in arbitrary order adding a vertex to the
  resulting set if and only if no previous neighboring vertex has been
  added.  In this loop, as in many sequential loops, each iterate will
  only depend on a subset of the previous iterates (i.e. knowing that
  any one of a vertex's neighbors is in the MIS, or knowing that it
  has no previous neighbors, is sufficient to decide its fate one way
  or the other).  This leads to a dependence structure among the
  iterates.  If this structure is shallow then running the iterates in
  parallel while respecting the dependencies can lead to an efficient
  parallel implementation mimicking the sequential algorithm.

  In this paper, we show that for any graph, and for a random ordering
  of the vertices, the dependence length of the sequential greedy MIS
  algorithm is polylogarithmic ($O(\log^2 n)$ with high probability).  Our results
  extend previous results that show polylogarithmic bounds only for
  random graphs.  We show similar results for a greedy maximal
  matching (MM).  For both problems we describe simple linear work
  parallel algorithms based on the approach.  The algorithms allow for
  a smooth tradeoff between more parallelism and reduced work, but
  always return the same result as the sequential greedy algorithms.
  We present experimental results that demonstrate efficiency and the
  tradeoff between work and parallelism.
\end{abstract}

\bigskip

\centerline{{\bf Keywords}: Parallel algorithms, maximal independent set, maximal matching}

\end{titlepage}

\section{Introduction}

The \defn{maximal independent set} (MIS) problem is given an undirected graph
$G = (V,E)$ to return a subset $U\subseteq V$ such that no vertices in
$U$ are neighbors of each other (independent set), and all vertices in
$V \setminus U$ have a neighbor in $U$ (maximal).  The MIS is a fundamental
problem in parallel algorithms with many applications~\cite{Luby86}.
For example if the vertices represent tasks and each edge represents
the constraint that two tasks cannot run in parallel, the MIS finds a maximal set
of tasks to run in parallel.
Parallel algorithms for the problem have been well
studied~\cite{KW84,Luby86,AlonBaIt86,Goldberg86,GoldbergPlSh87,GoldbergSp89a,GoldbergSp89b,CoppersmithRaTo89,CalkinFr90}.
Luby's randomized algorithm~\cite{Luby86}, for example, runs in
$O(\log |V|)$ time on $O(|E|)$ processors of an CRCW PRAM and
can be converted to run in linear work.  The problem, however, is that
on a modest number of processors it is very hard for these parallel
algorithms to outperform the very simple and fast sequential greedy
algorithm.  Furthemore the parallel algorithms give different results
than the sequential algorithm.  This can be undesirable in a context
where one wants to choose between the algorithms based on platform but
wants deterministic answers.

In this paper we show that, perhaps surprisingly, a trivial
parallelization of the sequential greedy algorithm is in fact highly
parallel (polylogarithmic time) when the order of vertices is randomized.
In particular simply by processing each vertex as in the sequential
algorithm but as soon as it has no earlier neighbors gives a parallel
linear work algorithm.  The MIS returned by the sequential greedy
algorithm, and hence also its parallelization, is referred to as the
\defn{lexicographically first} MIS~\cite{Cook85}.  In a general
undirected graph and an arbitrary order, the problem of finding a
lexicographically first MIS is
P-complete~\cite{Cook85,GreenlawHoRu95}, meaning that it is unlikely
that any efficient low-depth parallel algorithm exists for this
problem.\footnote{ Cook~\cite{Cook85} shows this for problem of
  lexicographically first maximal clique, which is equivalent to
  finding the MIS on the complement graph.}  Moreover, it is even
P-complete to approximate the size of the lexicographically first
MIS~\cite{GreenlawHoRu95}.  Our results show that for any graph and
for the vast majority of orders the lexicographically first MIS has
polylogarithmic depth.

Beyond theoretical interest the result has important practical
implications.  Firstly it allows for a very simple and efficient
parallel implementation of MIS that can trade off work with depth.
The approach is instead of processing all vertices in parallel to
process prefixes of the vertex ordering in parallel.  Using smaller
prefixes reduces parallelism but also reduces redundant work.  The
limit of a prefix of size one gives the sequential algorithm with no
redundant work.  We show bounds on prefix size that guarantee linear
work.  The second implication is that once an ordering is fixed, the
approach guarantees the same result whether run in parallel or
sequentially or, in fact, choosing any schedule of the iterations that
respects the dependences.  Such determinism can be an important
property of parallel algorithms~\cite{Bocchino09,BFGS12}.

Our results generalize the work of Coppersmith et
al.~\cite{CoppersmithRaTo89} (CRT) and Calkin and
Frieze~\cite{CalkinFr90} (CF).  CRT provide a greedy parallel
algorithm for finding a lexicographically first MIS for a random graph
$G_{n,p}$, $0 \leq p \leq 1$, where there are $n$ vertices and the
probability that an edge exists between any two vertices is $p$.  It
runs in $O(\log^2n / \log\log n)$ expected depth on a linear number of
processors.  CF give a tighter analysis showing that this algorithm
runs in $O(\log n)$ expected depth.  They rely heavily on the fact
that edges in a random graph are uncorrelated, which is not the case
for general graphs so their results do not extend to our context.  We
however use a similar approach of analyzing prefixes of the sequential
ordering.

The \defn{maximal matching} (MM) problem is given an undirected graph $G =
(V,E)$ to return a subset $E'\subseteq E$ such that no edges in $E'$
share an endpoint, and all edges in $E \setminus E'$ have a neighboring edge
in $E'$.  The MM of $G$ can be solved by finding an MIS of its line
graph (the graph representing adjacencies of edges in $G$), but the
line graph can be asymptotically larger than $G$.  Instead, the
efficient (linear time) sequential greedy algorithm goes through the
edges in an arbitrary order adding an edge if no adjacent edge has
already been added.  As with MIS this is naturally parallelized by
adding in parallel all edges that have no earlier neighboring edges.
Our results on MIS directly imply that this algorithm has
polylogarithmic depth for random edge ordering with high probability.  We also show
how to make the algorithm linear work, which requires modifications to
the MIS algorithm.  Previous work has also shown polylogarithmic depth
and linear work algorithms for the MM
problem~\cite{IsraeliSh86,IsraeliIt86} but as with MIS ours approach
returns the same results as the sequential algorithm and leads to very
efficient code.

Our experiments show how the choice of prefix size affects total work performed, parallelism, and overall running time. With a careful choice of prefix size, our algorithms indeed achieve very good speed-up (14--24x on 32 processors) and require only a modest number of processors to outperform optimized sequential implementations. Our efficient implementation of Luby's algorithm requires many more processors to outperform its sequential counterpart. Our prefix-based MIS algorithm is always 4--8 times faster than the implementation of Luby's algorithm, since our prefix-based algorithm performs less work in practice.

\section{Notation and Preliminaries}
Throughout the paper, we use $n$ and $m$ to refer to the number of
vertices and edges, respectively, in the graph. For a graph $G =
(V,E)$ we use $N_G(V)$ (or $N(V)$ when clear) to denote the set of all
neighbors of vertices in $V$, and $N_G(E)$ the neighboring edges of
$E$ (ones that share a vertex).  A maximal independent set $U \subset
V$ is thus one that satisfies $N_G(U) \cap U = \emptyset$ and $N_G(U)
\cup U = V$, and a maximal matching $E'$ is
one that satisfies $N(E') \cap E' = \emptyset$ and $N(E') \cup E' =
E$. We use $N(v)$ as a shorthand for $N(\set{v})$ when $v$ is a single
vertex.  We use $G[U]$ to denote the \defn{vertex-induced subgraph} of
$G$ by vertex set $U$, i.e., $G[U]$ contains all vertices in $U$ along
with edges with both endpoints in $U$.  We use $G[E']$ to denote the
\defn{edge-induced subgraph} of $G$, i.e., $G[E']$ contains all edges
$E'$ along with the incident vertices of $G$.

Throughout this paper, we use the concurrent-read concurrent-write
(CRCW) parallel random access machine (PRAM) model for analyzing
algorithms.  We assume arbitrary write version.  Our results are
stated in the work-depth model where work is equal to the number of
operations (equivalently the product of the time and processors) and
depth is equal to the number of time steps.

\newcommand{\pdag}{priority DAG}
\newcommand{\ddepth}{dependence length}
\newtheorem{corollary}[theorem]{Corollary}

\algblock{ParFor}{EndParFor}
\algnewcommand\algorithmicparfor{\textbf{parfor}}
\algnewcommand\algorithmicpardo{\textbf{do}}
\algnewcommand\algorithmicendparfor{\textbf{end\ parfor}}
\algrenewtext{ParFor}[1]{\algorithmicparfor\ #1\ \algorithmicpardo}
\algrenewtext{EndParFor}{\algorithmicendparfor}

\section{Maximal independent set}


The sequential algorithm for computing the MIS of a graph is a simple
greedy algorithm, shown in Algorithm~\ref{alg:sequential}.  In
addition to a graph $G$ the algorithm takes an arbitrary total
ordering on the vertices $\pi$.  We also refer to $\pi$ as priorities
on the vertices.  The algorithm adds the first remaining vertex $v$
according to $\pi$ to the MIS and then removes $v$ and all of $v$'s
neighbors from the graph, repeating until the graph is empty.  The MIS
returned by this sequential algorithm is defined as the
lexicographically first MIS for $G$ according to $\pi$.  

\begin{algorithm} 
\caption{Sequential greedy algorithm for maximal independent set} \label{alg:sequential}
\begin{algorithmic}[1]
\Procedure{Sequential-Greedy-MIS}{$G = (V,E)$, $\pi$ }
\If{ $|V| = 0$ } \Return $\O$
\Else \State let $v$ be the first vertex in $V$ by the ordering $\pi$
\State $V' = V \setminus (v \cup N(v))$
\State \Return $v$ $\cup$ \Call{Sequential-Greedy-MIS}{$G[V']$, $\pi$ }
\EndIf
\EndProcedure
\end{algorithmic}
\end{algorithm}

\vspace{-0.1in}

\begin{algorithm}
\caption{Parallel greedy algorithm for maximal independent set} \label{alg:fully-parallel}
\begin{algorithmic}[1]
\Procedure{Parallel-greedy-MIS}{$G = (V,E)$, $\pi$ }
\If{ $|V| = 0$ } \Return $\O$
\Else
\State let $W$ be the set of vertices in $V$ with no earlier neighbors (based on $\pi$)
\State $V' = V \setminus (W \cup  N(W))$
\State \Return $W$ $\cup$ \Call{Parallel-greedy-MIS}{ $G[V']$, $\pi$ }
\EndIf
\EndProcedure
\end{algorithmic}
\end{algorithm}

By allowing vertices to be added to the MIS as soon as they have no
earlier neighbor we get the parallel
Algorithm~\ref{alg:fully-parallel}.  It is not difficult to see that
this algorithm returns the same MIS as the sequential algorithm.  A
simple proof proceeds by induction on vertices in order.  (A vertex
$v$ may only be resolved when all of its earlier neighbors have been
classified.  If its earlier neighbors match the sequential algorithm,
then it does too.)  Naturally, the parallel algorithm may (and should,
if there is to be any speedup) accept some vertices into the MIS at an
earlier time than the sequential algorithm but the final set produced
is the same.  

We also note that if Algorithm~\ref{alg:fully-parallel} regenerates
the ordering $\pi$ randomly on each recursive call then the algorithm
is effectively the same as Luby's Algorithm A~\cite{Luby86}.  It is
the fact that we use the same permutation which makes this
Algorithm~\ref{alg:fully-parallel} more difficult to analyze.

\paragraph{The priority DAG.}
A perhaps more intuitive way to view this algorithm is in terms of a
directed acyclic graph (DAG) over the input vertices where edges are
directed from higher priority to lower priority endpoints based on
$\pi$.  We call this DAG the \defn{\pdag{}}. We refer to each
recursive call of Algorithm~\ref{alg:fully-parallel} as a
\defn{step}.  Each step adds the roots\footnote{We use the term
``root'' to refer to those nodes in a DAG with no incoming edges.} of
the \pdag{} to the MIS and removes them and their children from the
\pdag{}. This process continues until no vertices remain.
We define the number of iterations to remove all vertices from the
\pdag{} (equivalently, the number of recursive calls in
Algorithm~\ref{alg:fully-parallel}) as its \defn{\ddepth{}}.  The
\ddepth{} is upper bounded by the longest directed path in the \pdag{},
but in general could be significantly less.  Indeed for a complete
graph the longest directed path in the \pdag{} is $\Omega(n)$, but the
\ddepth{} is $O(1)$.

The main goal of this section is to show that the \ddepth{} is
polylogarithmic for most orderings~$\pi$.  Instead of arguing this
fact directly, we consider \pdag{}s induced by subsets of vertices and
show that these have small longest paths and hence small \ddepth{}.
Aggregating across all sub-DAGs gives an upper bound on the total
\ddepth{}.

\subsubsection*{Analysis via modified parallel algorithm}
Analyzing the depth of Algorithm~\ref{alg:fully-parallel} directly
seems difficult as once some vertices are removed, the ordering among the
set of remaining vertices may not be uniformly random.  Rather than
analyzing the algorithm directly, we preserve sufficient independence
over priorities by adopting an analysis framework similar
to~\cite{CoppersmithRaTo89,CalkinFr90}.  Specifically, for the purpose
of analysis, we consider a more restricted, less parallel algorithm
given by Algorithm~\ref{alg:prefix-parallel}.  

\begin{algorithm}
\caption{Modified parallel greedy algorithm for maximal independent set} \label{alg:prefix-parallel}
\begin{algorithmic}[1]
\Procedure{Modified-Parallel-MIS}{ $G=(V,E)$, $\pi$ }
\If{ $|V| = 0$ } \Return $\O$
\Else
\State choose prefix-size parameter $\delta$ \Comment{$\delta$ may be a function of
  $G$} 
\State let $P = P(V,\pi,\delta)$ be the vertices in the prefix
\State $W = $ \Call{Parallel-greedy-MIS}{ $G[P]$, $\pi$ }
\State $V' = V \setminus (P \cup N(W))$
\State \Return $W$ $\cup$ \Call{Modified-Parallel-MIS}{ $G[V']$, $\pi$ }
\EndIf
\EndProcedure
\end{algorithmic}
\end{algorithm}

Algorithm~\ref{alg:prefix-parallel} differs from
Algorithm~\ref{alg:fully-parallel} in that it considers only a prefix
of the remaining vertices rather than considering all vertices in
parallel.  This modification may cause some vertices to be processed
later than they would in Algorithm~\ref{alg:fully-parallel}, which can
only \emph{increase} the total number of steps of the algorithm.  We
will show that Algorithm~\ref{alg:prefix-parallel} has a
polylogarithmic number of steps, and hence
Algorithm~\ref{alg:fully-parallel} also does.

We refer to each iteration (recursive call) of
Algorithm~\ref{alg:prefix-parallel} as a \defn{round}.  For an ordered
set $V$ of vertices and fraction $0<\delta\leq 1$, we define the
\defn{$\delta$-prefix} of $V$, denoted by $P(V,\pi,\delta)$, to be the
subset of vertices corresponding to the $\delta\card{V}$ earliest in
the ordering $\pi$.  During each round, the algorithm selects the
$\delta$-prefix of remaining vertices for some value of $\delta$ to be
discussed later.  An MIS is then computed on the vertices in the prefix
using Algorithm~\ref{alg:prefix-parallel}, ignoring the rest of the
graph.  When the call to Algorithm~\ref{alg:prefix-parallel} finishes,
all vertices in the prefix have been processed and either belong to
the MIS or have a neighbor in the MIS.  All neighbors of these newly
discovered MIS vertices and their incident edges are removed from the
graph to complete the round.

The advantage of analyzing Algorithm~\ref{alg:prefix-parallel} instead
of Algorithm~\ref{alg:fully-parallel} is that at the beginning of each
round, the ordering among remaining vertices is still uniform, as the
only information the algorithm has discovered about each vertex is
that it appears after the last prefix. The goal of the analysis is
then to argue a) that the number of steps in each parallel round is small,
and that b) the number of rounds is small.  The latter can be
accomplished directly by selecting prefixes that are ``large enough,''
and constructively using a small number of rounds.  Larger prefixes
increase the number of steps within each round, however, so some care must be
taken in tuning the prefix sizes.

Our analysis assumes that the graph is arbitrary (i.e., adversarial),
but that the ordering on vertices is random.  In contrast, the
previous analysis in this style~\cite{CoppersmithRaTo89,CalkinFr90}
assume that the underlying graph is random, a fact which is exploited
to show that the number of steps within each round is small.   Our analysis, on
the other hand, must cope with nonuniformity on the permutations of
(sub)prefixes as the prefix is processed with
Algorithm~\ref{alg:fully-parallel}.

\subsubsection*{Reducing vertex degrees}
A significant difficulty in analyzing the number of steps of a single round of
Algorithm~\ref{alg:prefix-parallel} (i.e., the execution of
Algorithm~\ref{alg:fully-parallel} on a prefix) is that the steps of
Algorithm~\ref{alg:fully-parallel} are not independent given a single
random permutation that is not regenerated after each iteration.  The
dependence, however, arises partly due to vertices of drastically
different degree, and can be bounded by considering only vertices of
nearly the same degree during each round.  

Let $\Delta$ be the \emph{a priori} maximum degree in the graph.  We
will select prefix sizes so that after the $i$th round, all remaining
vertices have degree at most $\Delta/2^i$ with high probability
\footnote{We use ``with high probability'' (w.h.p.) to mean
    probability at least $1-1/n^c$ for any constant $c$, affecting the
    constants in order notation.}.
After $\log\Delta < \log n$ rounds, all vertices have degree 0, and
thus can be removed in a single step.  Bounding the number of steps in
each round to $O(\log n)$ then implies that
Algorithm~\ref{alg:prefix-parallel} has $O(\log^2 n)$ total steps, and
hence so does Algorithm~\ref{alg:fully-parallel}.

The following lemma and corollary state that after the first
$\Omega(n\log(n)/d)$ vertices have been processed, all remaining
vertices have degree at most $d$.

\begin{lemma}\label{lem:degrees}
  Suppose that the ordering on vertices is uniformly random, and
  consider the $(\ell/d)$-prefix for any positive $\ell$ and
  $d\leq n$.  If a lexicographically first MIS of the prefix and all
  of its neighbors are removed from $G$, then all remaining vertices
  have degree at most $d$ with probability at least $1-n/e^\ell$.
\end{lemma}
\begin{proof}
  Consider the following sequential process, equivalent to the
  sequential Algorithm~\ref{alg:sequential} (in this proof we will refer to a recursive call of Algorithm~\ref{alg:sequential} as a step).  The process consists of
  $\ell/d$ steps.  Initially, all vertices are \defn{live}.
  Vertices become \defn{dead} either when they are added to the MIS or
  when a neighbor is added to the MIS.  During each step, randomly
  select a vertex $v$, without replacement.  The selected vertex may
  be live or dead.  If $v$ is live, it has no earlier neighbors in the
  MIS.  Add $v$ to the MIS, after which $v$ and all of its neighbors
  become dead.  If $v$ is already dead, do nothing.  Since vertices
  are selected in a random order, this process is equivalent to
  choosing a permutation first then processing the prefix.

  Fix any vertex $u$.  We will show that by the end of this sequential
  process, $u$ is unlikely to have more than $d$ live neighbors.
  (Specifically, during each step that it has $d$ neighbors, it is
  likely to become dead; thus, if it remains live, it is unlikely to
  have many neighbors.)  Consider the $i$th step of the sequential
  process. If either $u$ is dead or $u$ has fewer than $d$ live
  neighbors, then the claim holds.  Suppose instead that $u$ has at
  least $d$ live neighbors.  Then the probability that the $i$th
  step selects one of these neighbors is at least $d/(n-i) >
  d/n$.  If the live neighbor is selected, that neighbor is added
  to the MIS and $u$ becomes dead.  The probability that $u$ remains
  live during this step is thus at most $1-d/n$.  Since each step
  selects the next vertex uniformly at random, the probability that no
  step selects any of the $d$ neighbors of $u$ is at most
  $(1-d/n)^{\delta n}$, where $\delta = \ell/d$. This failure
  probability  is at most $((1-d/n)^{n/d})^\ell < (1/e)^\ell$. Taking a union bound over all vertices completes the proof.
\end{proof}

\begin{corollary}\label{cor:degrees}
  Let $\Delta$ be the \emph{a priori} maximum vertex degree.  Setting
  $\delta=\Omega(2^i\log(n)/\Delta)$ for the $i$th round of
  Algorithm~\ref{alg:prefix-parallel}, all remaining vertices after
  the $i$th round have degree at most $\Delta/2^i$, with high
  probability.

\end{corollary}
\begin{proof}
  This follows from Lemma~\ref{lem:degrees} with $\ell =
  \Omega(\log(n))$ and $d=\Delta/2^i$.  
\end{proof}

\subsubsection*{Bounding the number of steps in each round}

To bound the \ddepth{} of each prefix in
Algorithm~\ref{alg:prefix-parallel}, we compute an upper bound on the
length of the longest path in the \pdag{} induced by the prefix, as
this path length provides an upper bound on the \ddepth{}.

The following lemma says that as long as the prefix is not too large
with respect to the maximum degree in the graph, then the longest path
in the \pdag{} of the prefix has length $O(\log n)$.  

\begin{lemma}\label{lem:prefix-depth}
  Suppose that all vertices in a graph have degree at most $d$, and
  consider a randomly ordered $\delta$-prefix.  For any $\ell$ and $r$
  with $\ell \geq r \geq 1$, if $\delta < r/d$, then the longest
  path in the \pdag{} has length $O(\ell)$ with probability at least
  $1-(r/\ell)^\ell$.
\end{lemma}
\begin{proof}
  Consider an arbitrary set of $k$ positions in the prefix---there
  are ${\delta n \choose k}$ of these, where $n$ is the number of
  vertices in the graph.\footnote{The number of vertices $n$ here
    refers to those that have not been processed yet. The bound holds
    whether or not this number accounts for the fact that some
    vertices may be ``removed'' from the graph out of order, as the
    $n$ will cancel with another term that also has the same
    dependence.}  Label these positions from lowest to highest
  $(x_1,\ldots,x_k)$.  To have a directed path in these positions,
  there must be an edge between $x_i$ and $x_{i+1}$ for $1 \leq i <
  k$.  For a random graph, the probability of each edge is
  independent, so the probability of a path is easy to bound.  For an
  arbitrary graph, however, the edge probabilities are not
  independent.  Having the prefix be randomly ordered is equivalent to
  first selecting a random vertex for position $x_1$, then $x_2$, then
  $x_3$, and so on.  The probability of an edge existing between $x_1$
  and $x_2$ is at most $d/(n-1)$, as $x_1$ has at most $d$ neighbors
  and there are $n-1$ other vertices remaining to sample from.  The
  probability of an edge between $x_2$ and $x_3$ then becomes at most
  $d/(n-2)$. (In fact, the numerator should be $d-1$ as $x_1$ already
  has an edge to $x_0$, but rounding up here only weakens the bound.)
  In general, the probability of an edge existing between $x_{i}$ and
  $x_{i+1}$ is at most $d/(n-i)$, as $x_{i}$ may have $d$ other
  neighbors and $n-i$ nodes remain in the graph.  The probability
  increases with each edge in the path since once $x_1,\ldots,x_{i}$
  have been fixed, we may know, for example, that $x_{i}$ has no edges
  to $x_0,\ldots,x_{i-2}$.  Multiplying the $k$ probabilities together
  gives us the probability of a directed path from $x_1$ to $x_k$,
  which we round up to $(d/(n-k))^k$.

  Taking a union bound over all ${\delta n \choose k}$ sets of $k$
  positions (i.e., over all length-$k$ paths through the prefix) gives
  us probability at most

  \begin{equation*}
    {\delta n \choose
      k}*(d/(n-k))^k \leq \left(\frac{e \delta n}{k}\right)^k
    *\left(\frac{d}{n-k}\right)^{k} 
    = \left(\frac{e\delta n d}{k(n-k)}\right)^k
    \leq \left(\frac{2e\delta d}{k}\right)^k
  \end{equation*}
  Where the last step holds for $k<n/2$. Setting $k = 2e\ell$ and
  $\delta < r/d$ gives a probability of at most $(r/\ell)^\ell$ of
  having a path of length $4e\ell$ or longer.  Note that if we have
  $4e\ell > n/2$, violating the assumption that $k < n/2$, then $n =
  O(\ell)$, and hence the claim holds trivially.
\end{proof}

\begin{corollary}\label{cor:logdepth}
  Suppose that all vertices in a graph have degree at most $d$, and
  consider a randomly ordered prefix.  
  For an $O(\log(n)/d)$-prefix or smaller, the longest path in the
  \pdag{} has length $O(\log(n))$ w.h.p.
  For a $(1/d)$-prefix or smaller, the longest path has length
  $O(\log(n)/\log\log(n))$ w.h.p.
\end{corollary}
\begin{proof}
  For the first claim, apply Lemma~\ref{lem:prefix-depth} with $\ell =
  2r = O(\log(n))$.  For the second claim, use $r = 1$ and $\ell =
  \log(n)/\log\log(n)$.
\end{proof}

Note that we want our bounds to hold with high probability with
respect to the original graph, so the $\log(n)$ in this corollary
should be treated as a constant across the execution of the algorithm.

\subsubsection*{Parallel greedy MIS has low \ddepth{}}
We now combine the fact that there are $\log n$ rounds with the
$O(\log n)$ steps per round to prove the following theorem on the
number of rounds in Algorithm~\ref{alg:fully-parallel}.  

\begin{theorem}\label{thm:numrounds}
  For a random ordering on vertices, where $\Delta$ is the maximum
  vertex degree, the \ddepth{} of the \pdag{} is $O(\log\Delta
  \log n) = O(\log^2 n)$ w.h.p.  Equivalently,
  Algorithm~\ref{alg:fully-parallel} requires $O(\log^2n)$ iterations
  w.h.p.
\end{theorem}
\begin{proof}
  We first bound the number of rounds of Algorithm~\ref{alg:prefix-parallel},
  choosing $\delta = c2^i\log(n)/\Delta$ in the $i$th round, for some
  constant $c$ and constant $\log(n)$ (i.e., $n$ here means the
  original number of vertices).  Corollary~\ref{cor:degrees} says that
  with high probability, vertex degrees decrease in each round.
  Assuming this event occurs (i.e., vertex degree is $d < \Delta/2^i$),
  Corollary~\ref{cor:logdepth} says that with high probability, the
  number of steps per round is $O(\log n)$.  Taking a union bound across any of
  these events failing says that every round decreases the degree
  sufficiently and thus the number of rounds required is $O(\log n)$ w.h.p.  We then
  multiply the number of steps in each round by the number of rounds to get the
  theorem bound.
  Since Algorithm~\ref{alg:prefix-parallel} only delays processing
  vertices as compared to Algorithm~\ref{alg:fully-parallel}, it
  follows that this bound on steps also applies to
  Algorithm~\ref{alg:fully-parallel}.
\end{proof}

\vspace{-0.1in}
\section{Achieving a linear work MIS algorithm}\label{sec:work}

While Algorithm~\ref{alg:fully-parallel} has low depth a na\"{\i}ve
implementation will require $O(m)$ work on each step to process all
edges and vertices and therefore a total $O(m \log^2 n)$ work.  Here
we describe two linear work versions.  The first is a smarter
implementation of Algorithm~\ref{alg:fully-parallel} that directly
traverses the \pdag{} only doing work on the roots and their neighbors
on each step---and therefore every edge is only processed once.  The
algorithm therefore does linear work and has computation depth that is
proportional to the
\ddepth{}.  The second follows the form of
Algorithm~\ref{alg:prefix-parallel}, only processing prefixes of
appropriate size.  It has the advantage that it is particularly easy
to implement.  We use this second algorithm for our experiments.

\subsubsection*{Linear work through maintaining root sets}

The idea of the linear work implementation of
Algorithm~\ref{alg:fully-parallel} is to explicitly keep on each step
of the algorithm the set of roots of the remaining
\pdag{}, e.g., as an array.  With this set it is easy to identify the
neighbors in parallel and remove them, but it is trickier to identify
the new root set for the next step.  One way to identify them would be
to keep a count for each vertex of the number of neighbors with higher
priorities (parents in the \pdag{}), decrement the counts whenever a
parent is removed, and add a vertex to the root set when its count
goes to zero.  The decrement, however, needs to be done in parallel
since many parents might be removed simultaneously.  Such decrementing
is hard to do work-efficiently when only some vertices are being
decremented.  Instead we note that the algorithm only needs to
identify which vertices have at least one edge removed on the step and
then check each of these to see if all their edges have been removed.
We refer to a \defn{misCheck} on a vertex as the operation of checking if
it has any higher priority neighbors remaining.  We assume the
neighbors of a vertex have been pre-partitioned into their parents
(higher priorities) and children (lower priorities), and that edges
are deleted lazily---i.e. deleting a vertex just marks it as deleted
without removing it from the adjacency lists of its neighbors.

\begin{lemma}\label{lem:check} 
  For a graph with $m$ edges and $n$ vertices where vertices are
  marked as deleted over time, any set of $l$ misCheck
  operations can be done in $O(l + m)$ total work, and each operation
  in $O(\log n)$ depth.
\end{lemma}
\begin{proof}
  The pointers to parents are kept as an array in an arbitrary order.
  A vertex can be checked by examining the parents in order.  If a
  parent is marked as deleted we remove the edge by incrementing the
  pointer to the array start and charging the cost to that edge.  If it
  is not, the misCheck completes and we charge the cost to the
  check.  Therefore the total we charge across all operations is $l +
  m$, each of which does constant work.  To implement this in parallel
  we use doubling: first examine one parent, then the next two, then
  the next four, etc.  This completes once we find one that is
  not deleted and within a factor of two we can charge all work to the
  previous ones that were deleted.  This requires $O(\log n)$ steps
  each with $O(1)$ depth. \end{proof}

\begin{lemma}\label{lem:worklist} 
  For a graph $G$ with $m$ edges and $n$ vertices
  Algorithm~\ref{alg:fully-parallel} can be implemented on a CRCW PRAM
  in $O(m)$ total work and $O(\log n)$ depth per iteration.
\end{lemma}
\begin{proof}
  The implementation works by keeping the roots in an array, and on
  each step marking the roots and its neighbors as deleted, and then
  using misCheck on the neighbors' neighbors to determine which
  ones belong in the root array for the next step.  The total number
  of checks is at most $m$, so the total work spent on checks is
  $O(m)$.  After the misCheck's all vertices with no previous vertex
  remaining are added to the root set for the next step.  Some care
  needs to be taken to avoid duplicates in the root array since
  multiple neighbors might check the same vertex.  Duplicates can be
  avoided, however, by having the neighbor write its identifier into
  the checked vertex using an arbitrary concurrent write, and
  whichever write succeeds is responsible for the check and adding the
  vertex to the new root array.  Each iteration can be implemented in
  $O(\log n)$ depth, required for the checks and for packing the
  successful checks into a new root set.  Every vertex and its edges
  are visited once when removing them, and the total work on checks is
  $O(m)$, so the overall work is $O(m)$.
\end{proof}

\subsubsection*{Linear work through smaller prefixes}

The naive algorithm has high work because it processes every vertex
and edge in every iteration.  Intuitively, if we process small-enough
prefixes (as in Algorithm~\ref{alg:prefix-parallel}) instead of the
entire graph, there should be less wasted work.  Indeed, a prefix of
size 1 yields the sequential algorithm with $O(m)$ work but
$\Omega(n)$ depth.  There is some tradeoff here---increasing the
prefix size increases the work but also increases the parallelism.
This section formalizes this intuition and describes a highly-parallel
algorithm that has linear work. 

To bound the work, we bound the number of edges operated on
while considering a prefix.  For any prefix $P \subseteq V$ with
respect to permutation $\pi$, we define \defn{internal edges} of
$P$ to be the edges in the sub-DAG induced by $P$, i.e., those
edges that connect vertices in $P$.  We call all other edges incident
on $P$ \defn{external edges}.   The internal edges may be processed
multiple times, but external edges are processed only once.

The following lemma states that small prefixes have few internal
edges.  We will use this lemma to bound the work incurred by
processing edges.  The important feature to note is that for very
small prefixes, i.e., $\delta < k/d$ with $k\ll 1$, the number of
internal edges in the prefix is sublinear in the size of the prefix,
so we can afford to process those edges multiple times.

\begin{lemma}\label{lem:fewedges}
  Suppose that all vertices in a graph have degree at most
  $d$, and consider a randomly ordered $\delta$-prefix $P$.  If
  $\delta < k/d$, then the expected number of internal edges in the
  prefix is at most $O(k \card{P})$. 
\end{lemma}
\begin{proof}
  Consider a particular vertex in $P$.  Each of its neighbors
  joins the prefix with probability $< k/d$, so the expected number
  of neighbors is at most $k$.  Summing over all vertices in $P$ gives the bound. 
\end{proof}

The following related lemma states that for small prefixes, most
vertices have no incoming edges and can be removed immediately.  We
will use this lemma to bound the work incurred by processing vertices,
even those that may have already been added to the MIS or implicitly
removed from the graph.  The same sublinearity applies here.

\begin{lemma}\label{lem:fewverts} Suppose that all vertices in a graph have degree at most
  $d$, and consider a randomly ordered $\delta$-prefix $P$.  If $\delta
  \leq k/d$, then the expected number of vertices in P with at least 1
  internal edge is at most $O(k\card{P})$.
\end{lemma}
\begin{proof}
  Let $X_E$ be the random variable denoting the number of internal
  edges in the prefix, and let $X_V$ be the random variable denoting
  the number of vertices in the prefix with at least 1 internal edge.
  Since an edge touches (only) two vertices, we have $X_V \leq 2X_E$.
  It follows that $E[X_V] \leq 2E[X_E]$, and hence $E[X_V] =
  O(k\card{P})$ from Lemma~\ref{lem:fewedges}.
\end{proof}

The preceding lemmas indicate that small-enough prefixes are very
sparse.  Choosing $k = 1/\log n$, for example, the expected size of
the subgraph induced by a prefix $P$ is $O(\card{P}/\log(n))$, and
hence it can be processed $O(\log n)$ times without exceeding linear
work.  This fact suggests the following theorem.  The implementation
given in the theorem is relatively simple.  The prefixes can be
determined as \emph{a priori} priority ranges, with lazy vertex status
updates. Moreover, each vertex and edge is only densely packed into a
new array once, with other operations being done in place on the
original vertex list.  

\begin{theorem}
  Algorithm~\ref{alg:prefix-parallel} can be implemented to run in
  $O(\log^4n)$ depth and expected $O(n+m)$ work on a common CRCW PRAM.
  The depth bound holds w.h.p.
\end{theorem}
\begin{proof}
  This implementation updates vertex status (entering the MIS or
  removed due to a neighbor) only when that vertex is part of a
  prefix.  

  Group the rounds into $O(\log n)$ \emph{superrounds}, each
  corresponding to an
  $O(\log(n)/d)$-prefix. Corollary~\ref{cor:degrees} states that all
  superrounds reduce the maximum degree sufficiently, w.h.p.  This
  prefix, however, may be too dense, so we divide each superround into
  $\log^2 n$ rounds, each operating on a $O((1/d)(1/\log(n))$-prefix
  $P$.  To implement a round, first process all external edges to
  remove those vertices with higher-priority MIS neighbors.  Then
  accept any remaining vertices with no internal edges into the MIS.
  These preceding steps are performed on the original vertex/edge
  lists, processing edges incident on the prefix a constant number of
  times.  Let $P' \subseteq P$ be the set of prefix vertices that
  remain at this point.  Use prefix sums to count the number of
  internal edges for each vertex (which can be determined by comparing
  priorities), and densely pack $G[P']$ into new arrays.  This packing
  has $O(\log n)$ depth and linear work.  Finally, process the induced
  subgraph $G[P']$ using a naive implementation of
  Algorithm~\ref{alg:fully-parallel}, which has depth $O(D)$ and work
  equal to $O(\card{G[P']} \cdot D)$, where $D$ is the \ddepth{} of $P'$.  From Corollary~\ref{cor:logdepth}, $D = O(\log n)$
  with high probability. Combining this with expected prefix size of
  $E[\card{G[P']}] = O(\card{P}/\log n)$ from
  Lemmas~\ref{lem:fewedges} and~\ref{lem:fewverts} yields expected
  $O(\card{P})$ work for processing the prefix.  Summing across all
  prefixes implies a total of $O(n)$ expected work for
  Algorithm~\ref{alg:fully-parallel} calls plus $O(m)$ work in the
  worst case for processing external edges.  Multiplying the $O(\log
  n)$ prefix depth across all $O(\log^3 n)$ rounds completes the proof
  for depth.
\end{proof}

\vspace{-0.15in}
\section{Maximal Matching}

One way to implement maximal matching (MM) is to reduce it to MIS by
replacing each edge with a vertex, and create an edge between all
adjacent edges.  This reduction, however, can significantly increase
the number of edges in the graph and therefore will not take work that
is linear in the size of the original graph.  Instead a standard
greedy sequential algorithm is to process the edges in an arbitrary
order and include the edge in the MM if and only if no neighboring
edge on either side has already been added.  
As with the vertices in
the greedy MIS algorithms, edges can be processed out of order when
they don't have any earlier neighboring edges.  This idea leads to
Algorithm~\ref{alg:parallel-mm} where $\pi$ is now an ordering of the
edges.

\begin{algorithm}
\caption{Parallel greedy algorithm for maximal matching} \label{alg:parallel-mm}
\begin{algorithmic}[1]
\Procedure{Parallel-greedy-MM}{ $G=(V,E)$, $\pi$ }
\If{ $|E| = 0$ } \Return $\O$
\Else
\State let $W$ be the set of edges in $E$ with no adjacent edges with higher priority by $\pi$
\State $E' = E \setminus (E \cup  N(E))$
\State \Return $W$ $\cup$ \Call{Parallel-greedy-MM}{ $G[E']$, $\pi$ }
\EndIf
\EndProcedure
\end{algorithmic}
\end{algorithm}
\vspace{-0.15in}

\begin{lemma}\label{thm:mm-numrounds}
  For a random ordering on edges, the number of rounds of
  Algorithm~\ref{alg:parallel-mm} is $O(\log^2 m)$ w.h.p.
\end{lemma}
\begin{proof}
This follows directly from the reduction to MIS described above.  In
particular an edge is added or deleted in
Algorithm~\ref{alg:parallel-mm} exactly on the same step it would be
for the corresponding MIS graph in Algorithm~\ref{alg:fully-parallel}.
Therefore Lemma~\ref{thm:numrounds} applies.
\end{proof}

We now show how to implement Algorithm~\ref{alg:parallel-mm}
in linear work.  
As with the algorithm used in Lemma~\ref{lem:worklist} we can maintain
on each round an array of roots (edges that have no neighboring edges
with higher priority) and use them to both delete edges and generate
the root set for the next round.  However, we cannot afford to look at
all the neighbors' neighbors.  Instead we maintain for each vertex an
array of its incident edges sorted by priority.  This is
maintained lazily such that deleting an edge only marks it as deleted
and does not immediately remove it from its two incident vertices.  We
say an edge is \defn{ready} if it has no remaining neighboring edges
with higher priority.  We use a \defn{mmcheck} procedure on a vertex to
which determine if any incident edge is ready and identifies the edge
if so---a vertex can have at most one ready incident edge.  We assume
mmchecks do not happen in parallel with marking edges as deleted.

\begin{lemma}\label{lem:check-mm}
  For a graph with $m$ edges and $n$ vertices where edges are
  marked as deleted over time, any set of $l$ mmcheck
  operations can be done in $O(l + m)$ total work, and each operation
  in $O(\log m)$ depth.
\end{lemma}
\begin{proof}
The mmcheck is partitioned into two phases.  The first identifies
the highest priority incident edge that remains, and the second checks
if that edge is also the highest priority on its other endpoint and
returns it if so.  The first phase can be done by scanning the edges
in priority order removing those that have been deleted and stopping
when the first non-deleted edge is found.  As in Lemma~\ref{lem:check}
this can be done in parallel using doubling in $O(\log m)$ depth, and
the work can be charged either to a deleted edge, which is removed, or
the check itself.  The total work is therefore $O(l + m)$.  The second
phase can similarly use doubling to see if the highest priority edge
is also the highest priority on the other side.
\end{proof}

\begin{lemma}\label{lem:work1} 
  Given a graph with $m$ edges, $n$ vertices, and a random permutation on the edges
  $\pi$, Algorithm~\ref{alg:parallel-mm} can be implemented on a CRCW
  PRAM in $O(m)$ total work and $O(\log^3 m)$ depth with high
  probability.
\end{lemma}
\begin{proof}
  Since the edge priorities are selected at random, the initial sort
  to order the edges incident on each vertex can be done in $O(m)$
  work and within our depth bounds w.h.p. using bucket
  sorting~\cite{CLRS}.  Initially the set of ready
  edges are selected by a using an mmcheck on all edges.  On each
  step of Algorithm~\ref{alg:parallel-mm} we delete the set of
  ready edges and their neighbors (by marking them), and then check
  all vertices incident on the far end of each of the deleted
  neighboring edges.  This returns the new set of ready edges in
  $O(\log m)$ depth.  Redundant edges can easily be removed.  Thus the
  depth per step is $O(\log m)$ and by
  Lemma~\ref{thm:mm-numrounds} the total depth is $O(\log^3 m)$.
  Every edge is deleted once and the total number of checks is $O(m)$, so the total work is $O(m)$.
\end{proof}

\vspace{-0.15in}
\section{Experiments}
We performed experiments of our algorithms using varying prefix sizes,
and show how prefix size affects work, parallelism, and overall
running time. We also compare the performance of our prefix-based algorithms with sequential implementations and additionally for MIS we compare with an implementation of Luby's algorithm. For our experiments we use two graph inputs---a sparse
random graph with $10^7$ vertices and $5 \times 10^7$ edges and an rMat
graph with $2^{24}$ vertices and $5 \times 10^7$ edges. The rMat
graph~\cite{ChakrabartiZF04} has a power-law distribution of degrees.

We ran our experiments on a 32-core (with hyper-threading) Dell
PowerEdge 910 with $4\times 2.26\mbox{GHZ}$ Intel 8-core X7560 Nehalem
Processors, a 1066MHz bus, and 64\mbox{GB} of main memory. The
parallel programs were compiled using the \texttt{cilk++} compiler
(build 8503) with the \texttt{-O2} flag.  The sequential programs were
compiled using \texttt{g++} 4.4.1 with the \texttt{-O2} flag.

For both MIS and MM, we observe that, as expected, increasing the
prefix size increases both the total work performed
(Figures~\ref{fig:misWorkRand},~\ref{fig:misWorkRMat},~\ref{fig:mmWorkRand}
and~\ref{fig:mmWorkRMat}) and the parallelism, which is estimated by
the number of rounds of the outer loop (selecting prefixes) the
algorithm takes to complete
(Figures~\ref{fig:misRoundsRand},~\ref{fig:misRoundsRMat},~\ref{fig:mmRoundsRand}
and~\ref{fig:mmRoundsRMat}). As expected, the total work performed and
the number of rounds taken by a sequential implementation are both
equal to the input size. By examining the graphs of running time
vs. prefix size
(Figures~\ref{fig:misTimeRand},~\ref{fig:misTimeRMat},~\ref{fig:mmTimeRand}
and~\ref{fig:mmTimeRMat}) we see that there is some optimal prefix
size between 1 (fully sequential) and the input size (fully parallel).
In the running time vs. prefix size graphs, there is a small bump when
the prefix-to-input size ratio is between $10^{-6}$ and $10^{-4}$
corresponding to the point when the for-loop in our implementation
transitions from sequential to parallel (we used a grain size of 256
for our loops).

We also compare our prefix-based algorithms to optimized sequential
implementations, and additionally for MIS we compare with an optimized
implementation of Luby's algorithm. We tried different implementations of Luby's algorithm and report the times for the fastest one. For MIS, our prefix-based
implementation using the optimal prefix size obtained from experiments
(see Figures~\ref{fig:misTimeRand} and~\ref{fig:misTimeRMat})
is 4--8 times faster than Luby's algorithm (shown in
Figures~\ref{fig:misPrefixVsProcsRand}
and~\ref{fig:misPrefixVsProcsRMat}), which essentially processes the
entire input as a prefix (along with reassigning the priorities of
vertices between rounds which the deterministic prefix-based algorithm
does not do). This demonstrates that our prefix-based approach,
although sacrificing some parallelism, leads to less overall work and
lower running time. When using more than 2 processors, our prefix-based
implementation of MIS outperforms the serial version, while our
implementation of Luby's algorithm requires 16 or more processors to
outperform the serial version. The prefix-based algorithm achieves 14--17x speedup on 32 processors. For MM, our prefix-based algorithm
outperforms the corresponding serial implementation with 4 or more
processors and achieves 21--24x speedup on 32 processors (Figures~\ref{fig:mmPrefixVsProcsRand}
and~\ref{fig:mmPrefixVsProcsRMat}). We note that since the serial MIS
and MM algorithms are so simple, it is not easy for a parallel
implementation to outperform the corresponding serial implementation.

\section{Conclusion}
We have shown that the ``sequential'' greedy algorithms for MIS and MM have polylogarithmic depth, for randomly ordered inputs (vertices for MIS and edges for MM). This gives random lexicographically first solutions for both of these problems, and in addition has important practical implications such as giving faster implementations and guaranteeing determinism. Our prefix-based approach leads to a smooth tradeoff between parallelism and total work and by selecting a good prefix size, we show experimentally that indeed our algorithms achieve very good speedup and outperform their serial counterparts using only a modest number of processors.

We believe that our approach can be applied to sequential greedy algorithms for other problems (e.g. spanning forest) and this is a direction for future work. An open question is whether the dependence length of our algorithms can be improved to $O(\log n)$.

\begin{figure*}[h]
  \centering
  \subfigure[Total work done vs. prefix size on a {\bf random graph} 
    ($n = 10^7, m = 5 \times 10^7$) ]{
    \hspace{-.015\linewidth}%
    \includegraphics[width=0.35\linewidth]{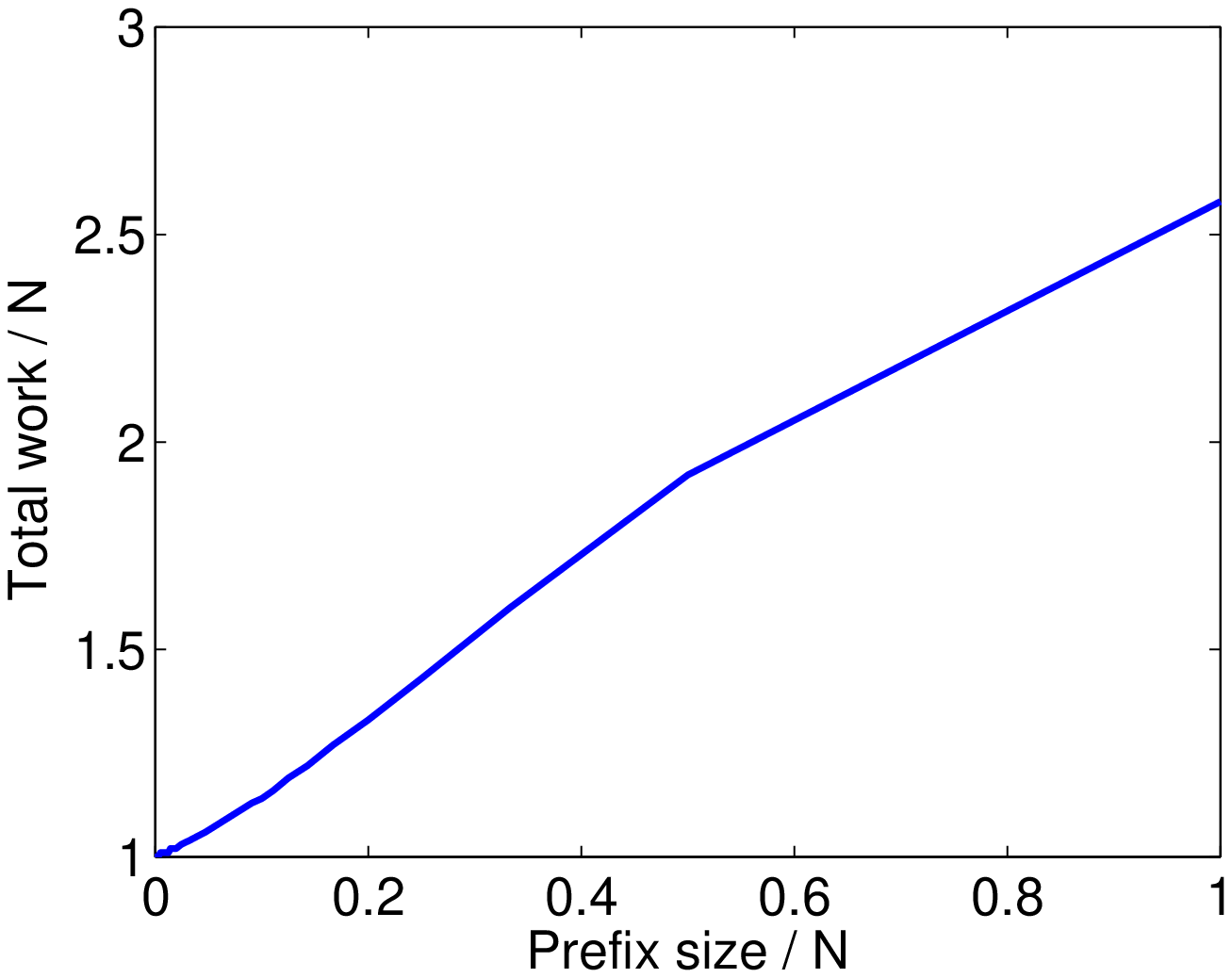}
    \hspace{-.03\linewidth}
    \label{fig:misWorkRand}
  } \hfill
  \subfigure[Number of rounds vs. prefix size on a {\bf random graph} 
     ($n = 10^7, m = 5 \times 10^7$) in log-log scale]{
    \hspace{-.015\linewidth}%
    \includegraphics[width=0.35\linewidth]{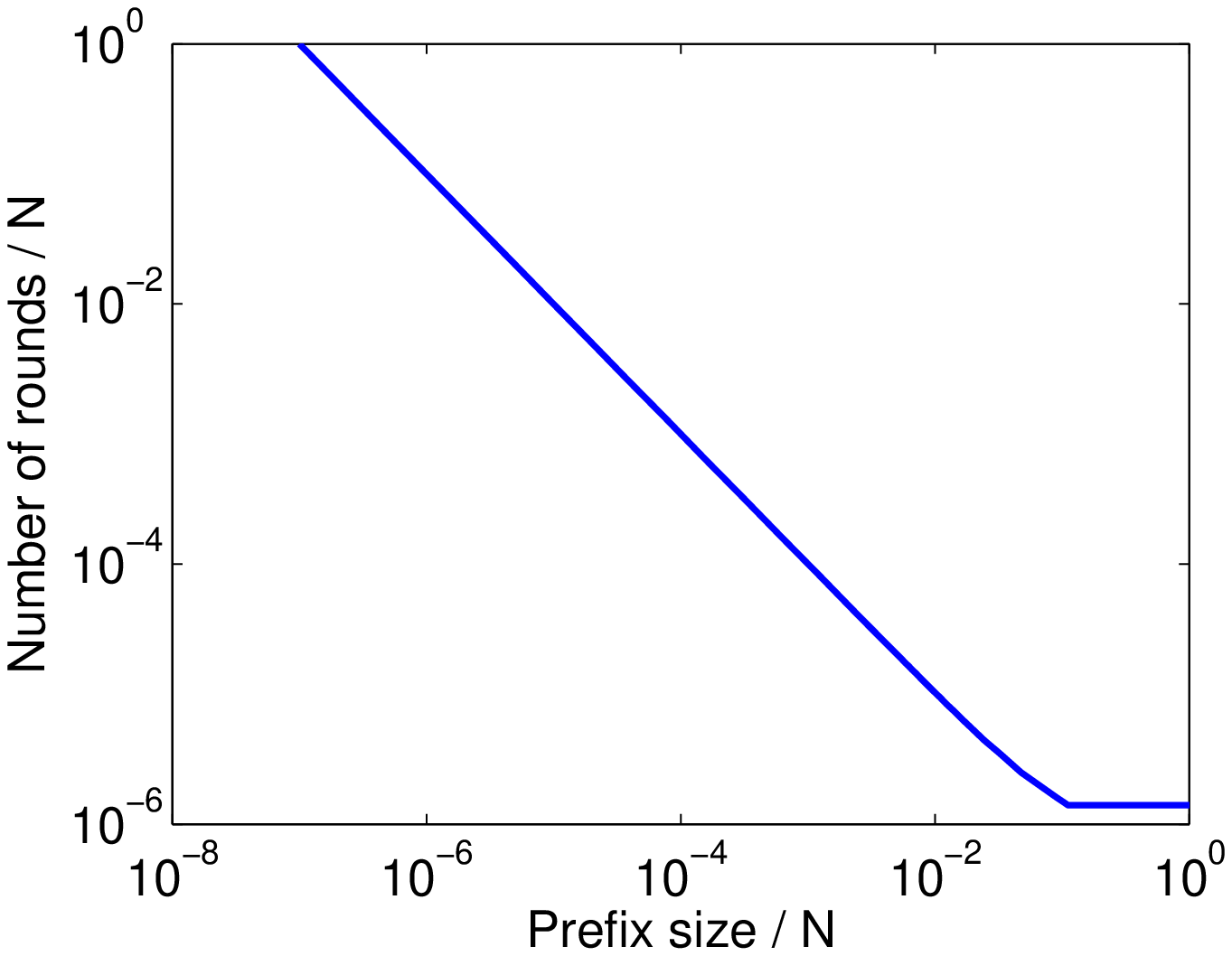}
    \hspace{-.03\linewidth}
    \label{fig:misRoundsRand}
  } \hfill
  \subfigure[Running time (32 processors) vs. prefix size on a {\bf random graph} 
	   ($n = 10^7, m = 5 \times 10^7$) in log-log scale]{
    \hspace{-.015\linewidth}%
    \includegraphics[width=0.35\linewidth]{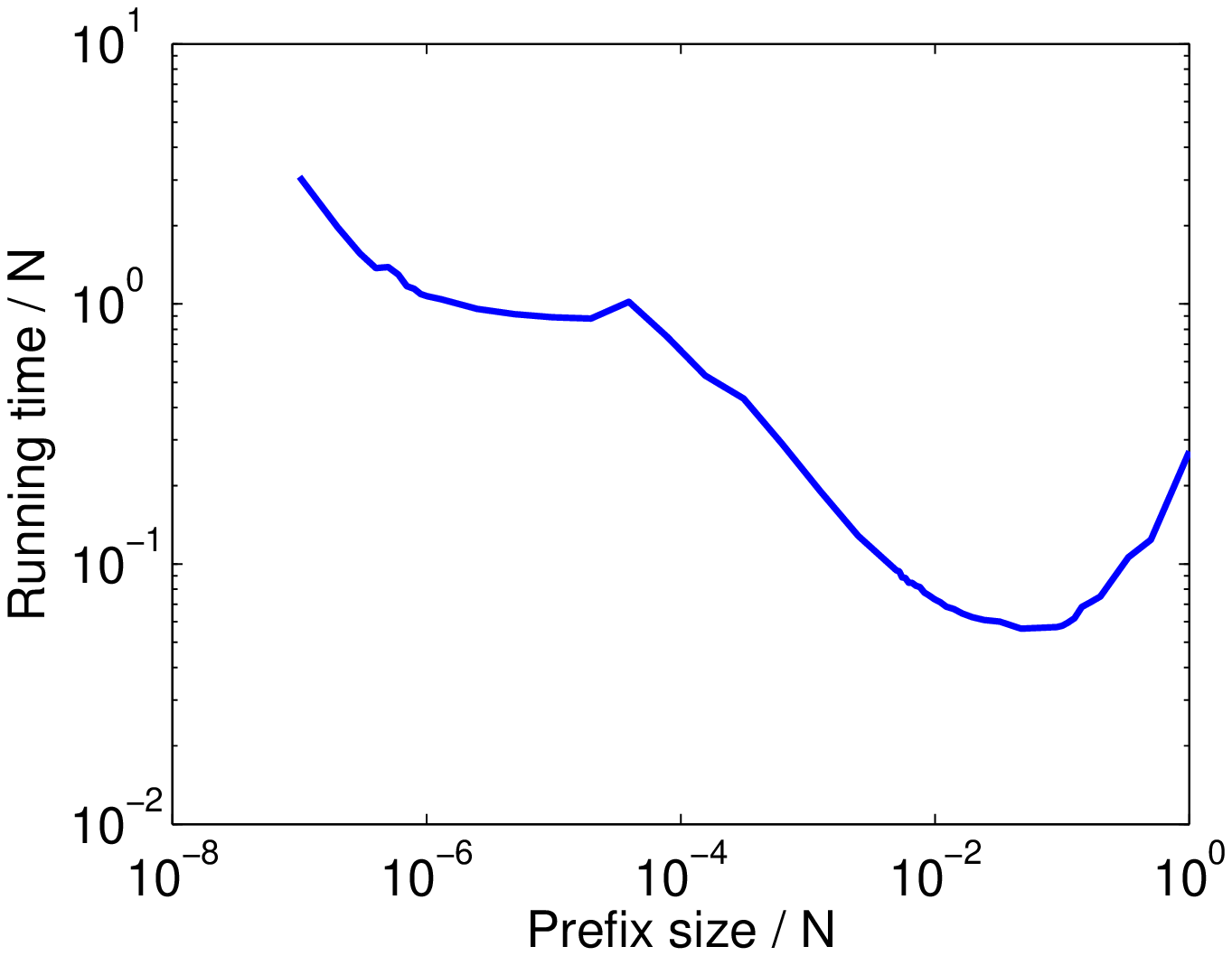}
    \hspace{-.03\linewidth}
    \label{fig:misTimeRand}
  }
  \subfigure[Total work done vs. prefix size on a {\bf rMat graph} ($n = 2^{24}, m = 5 \times 10^7$)]{
    \hspace{-.015\linewidth}%
    \includegraphics[width=0.35\linewidth]{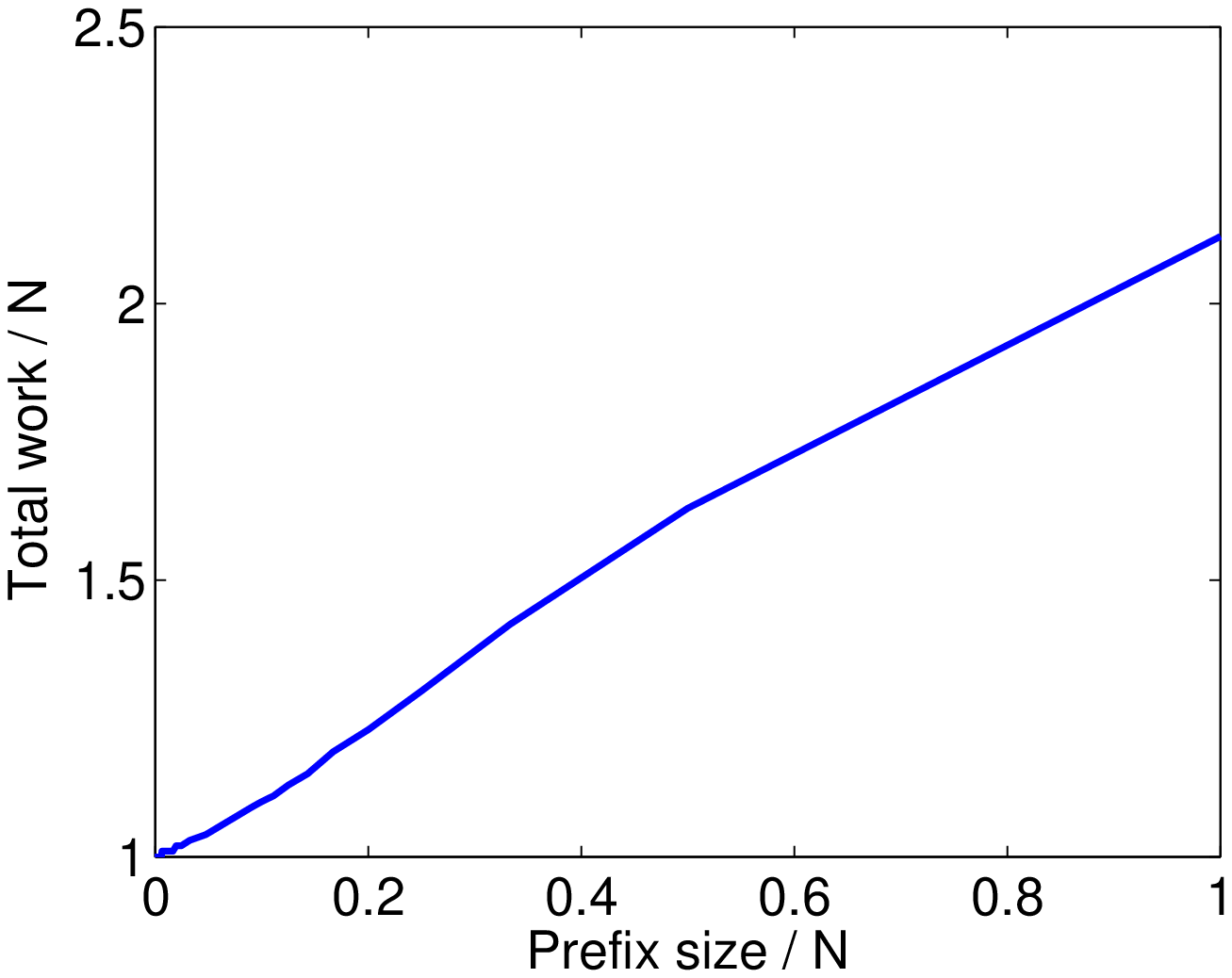}
    \hspace{-.03\linewidth}
    \label{fig:misWorkRMat}
  } \hfill
  \subfigure[Number of rounds vs. prefix size on a {\bf rMat graph} ($n = 2^{24}, m = 5 \times 10^7$) in log-log scale]{
    \hspace{-.015\linewidth}%
    \includegraphics[width=0.35\linewidth]{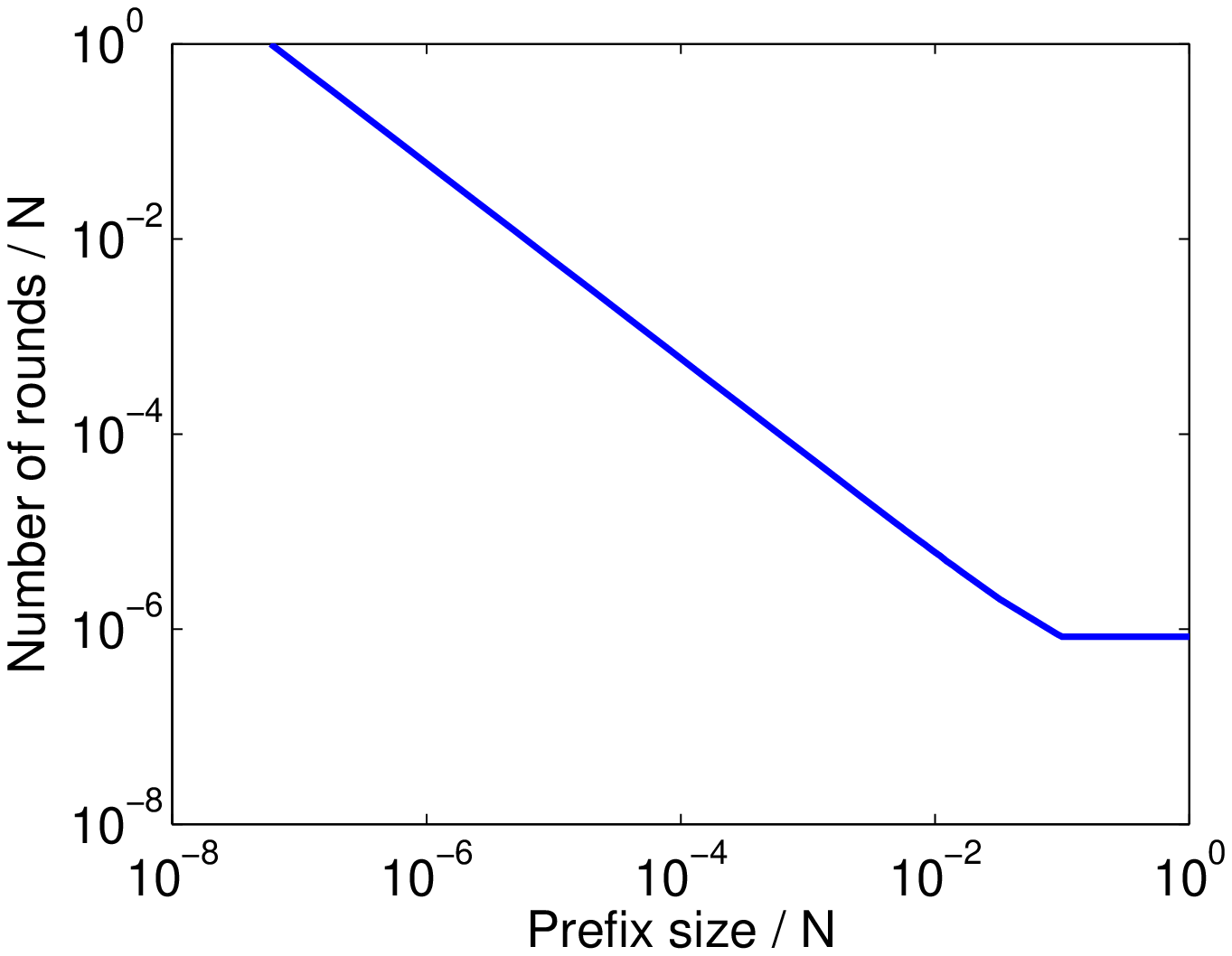}
    \hspace{-.03\linewidth}
    \label{fig:misRoundsRMat}
  } \hfill
  %
  %
  \subfigure[Running time (32 processors) vs. prefix size on a {\bf rMat graph} ($n = 2^{24}, m = 5 \times 10^7$) in log-log scale]{
    \hspace{-.015\linewidth}%
    \includegraphics[width=0.35\linewidth]{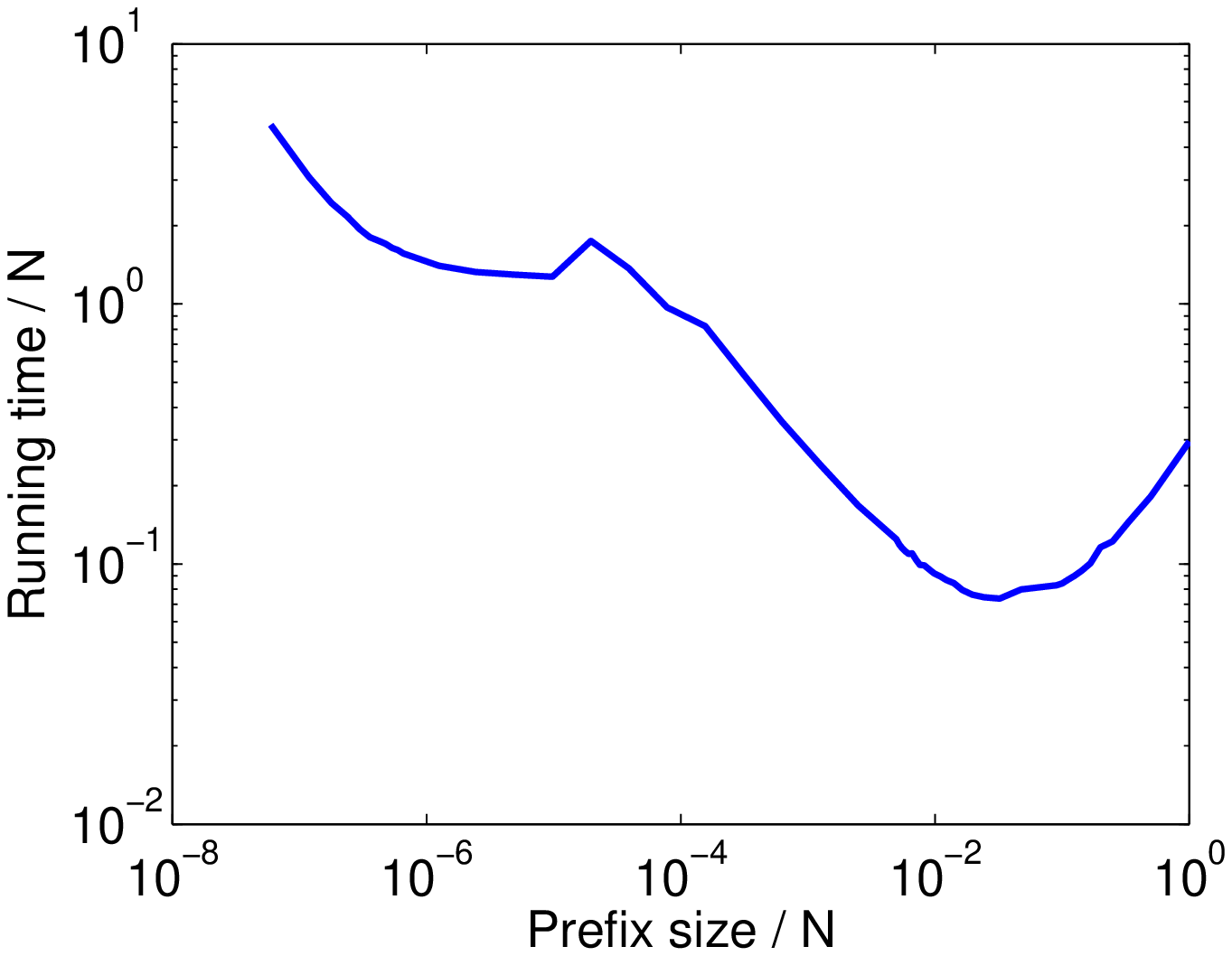}
    \hspace{-.03\linewidth}
    \label{fig:misTimeRMat}
  }
  \caption{Plots showing the tradeoff between various properties and the prefix size in maximal independent set.}
\end{figure*}

\begin{figure*}[h]
  \centering
  \subfigure[Total work done vs. prefix size on a {\bf random graph} 
    ($n = 10^7, m = 5 \times 10^7$) ]{
    \hspace{-.015\linewidth}%
    \includegraphics[width=0.35\linewidth]{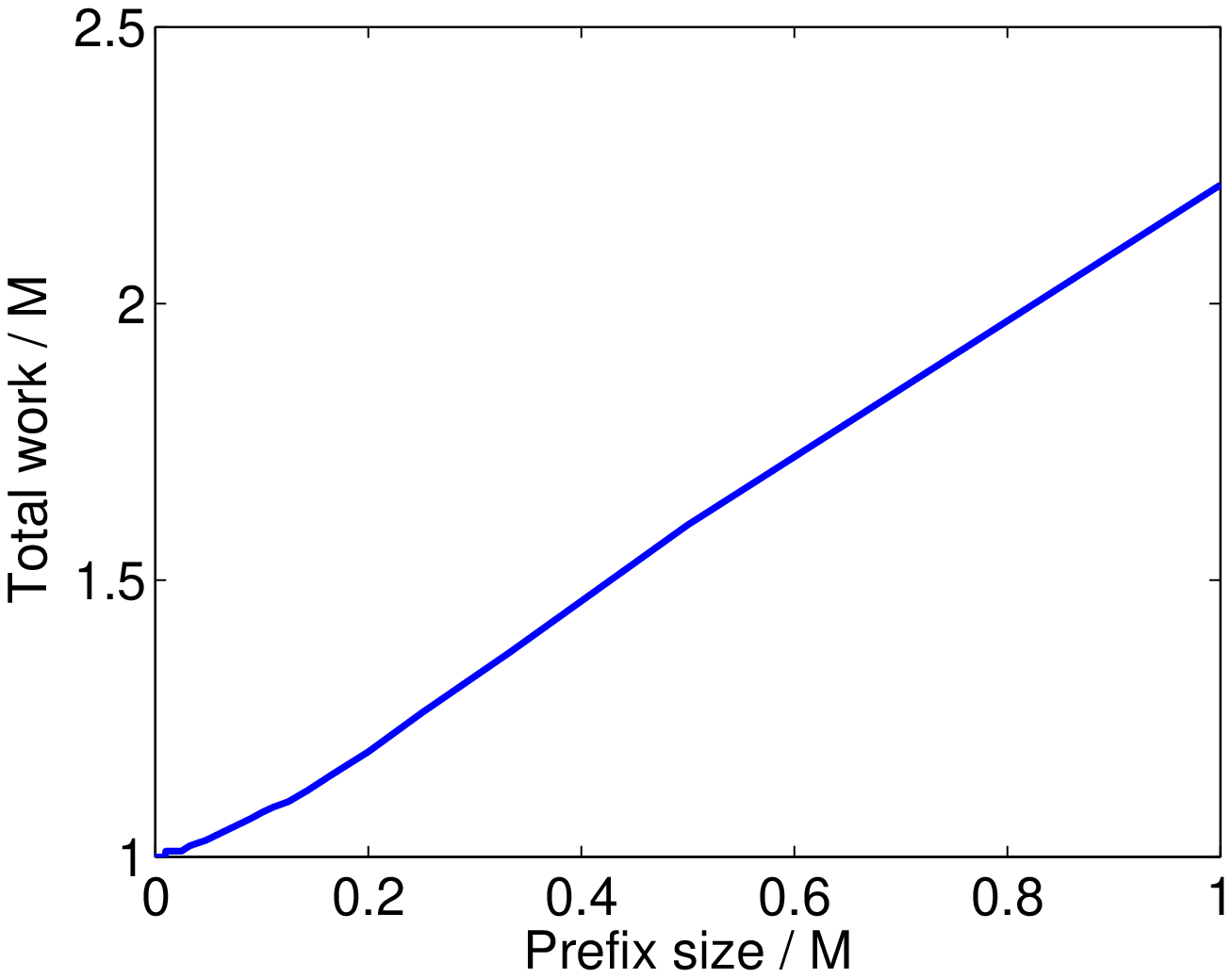}
    \hspace{-.03\linewidth}
    \label{fig:mmWorkRand}
  } \hfill
  \subfigure[Number of rounds vs. prefix size on a {\bf random graph} 
     ($n = 10^7, m = 5 \times 10^7$) in log-log scale]{
    \hspace{-.015\linewidth}%
    \includegraphics[width=0.35\linewidth]{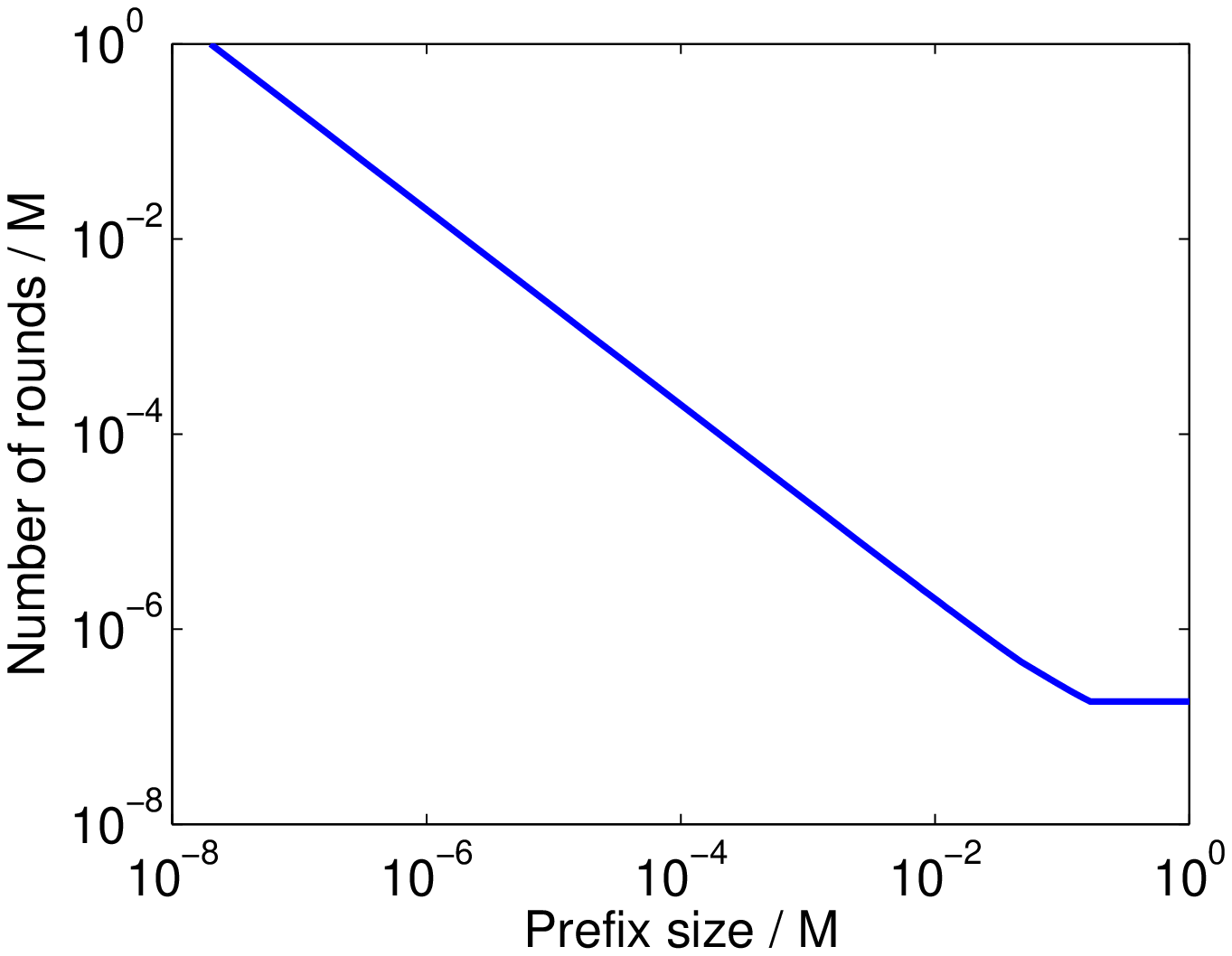}
    \hspace{-.03\linewidth}
    \label{fig:mmRoundsRand}
  } \hfill
  \subfigure[Running time (32 processors) vs. prefix size on a {\bf random graph} 
	   ($n = 10^7, m = 5 \times 10^7$) in log-log scale]{
    \hspace{-.015\linewidth}%
    \includegraphics[width=0.35\linewidth]{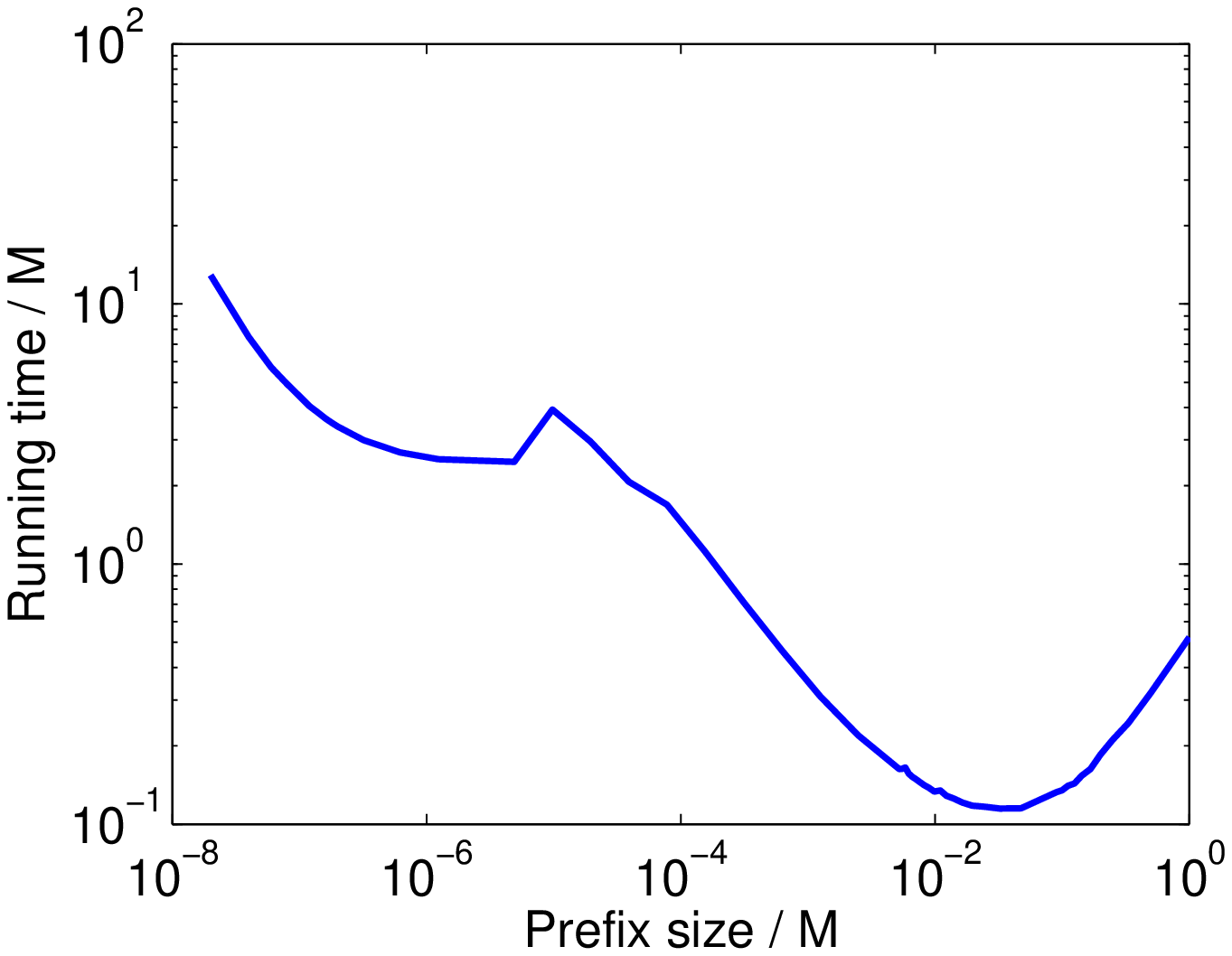}
    \hspace{-.03\linewidth}
    \label{fig:mmTimeRand}
  }
  \subfigure[Total work done vs. prefix size on a {\bf rMat graph} ($n = 2^{24}, m = 5 \times 10^7$)]{
    \hspace{-.015\linewidth}%
    \includegraphics[width=0.35\linewidth]{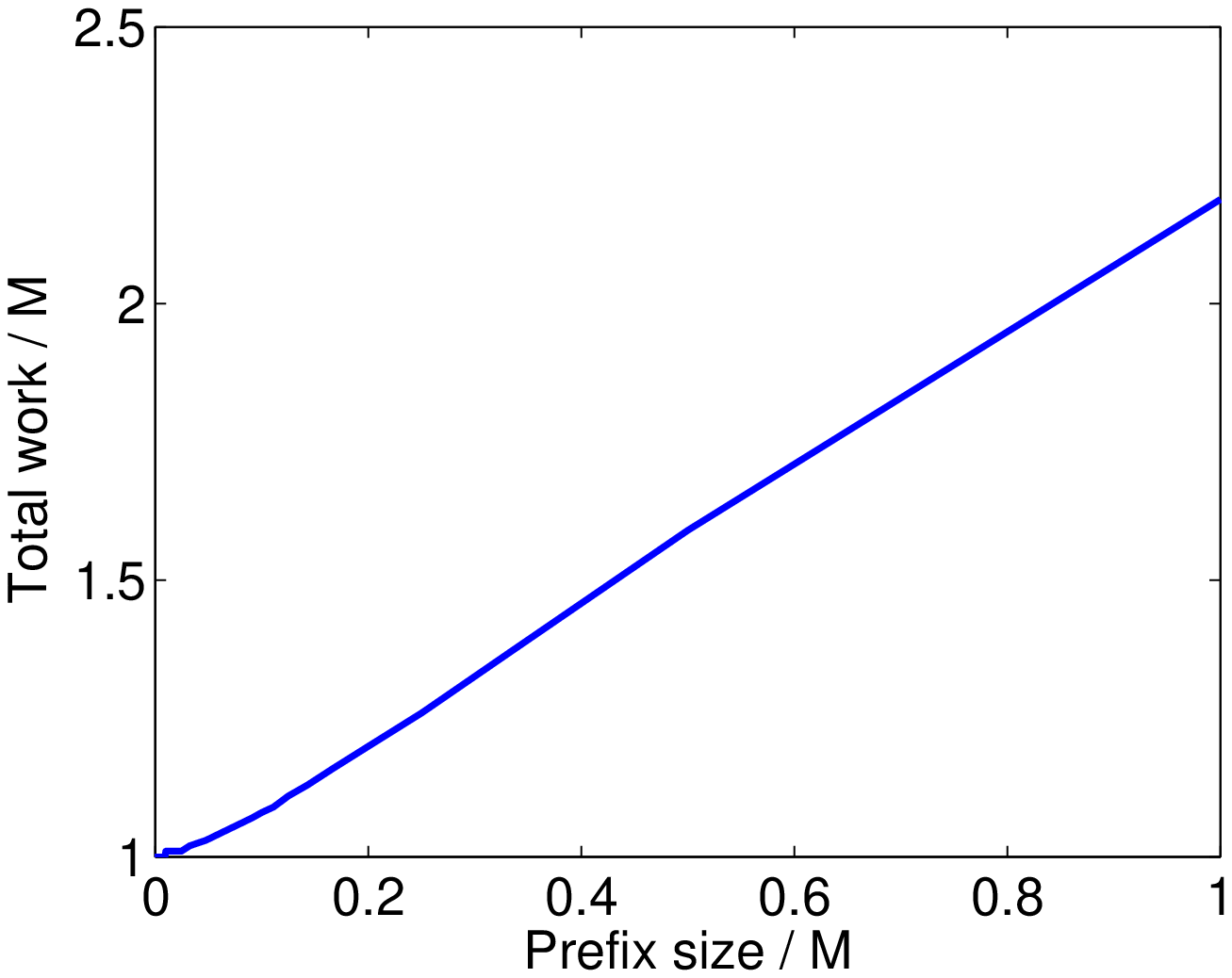}
    \hspace{-.03\linewidth}
    \label{fig:mmWorkRMat}
  } \hfill
  \subfigure[Number of rounds vs. prefix size on a {\bf rMat graph} ($n = 2^{24}, m = 5 \times 10^7$) in log-log scale]{
    \hspace{-.015\linewidth}%
    \includegraphics[width=0.35\linewidth]{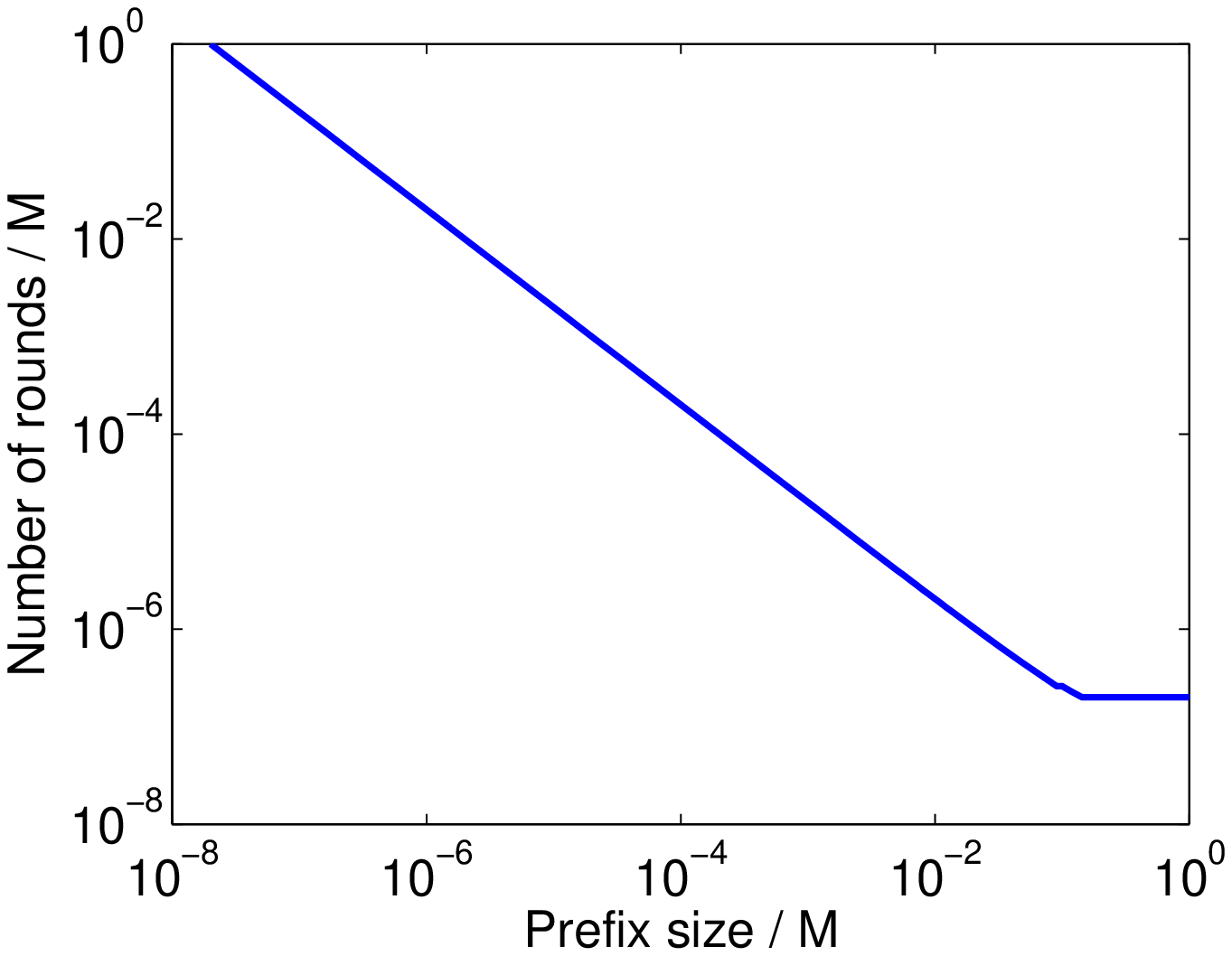}
    \hspace{-.03\linewidth}
    \label{fig:mmRoundsRMat}
  } \hfill
  %
  %
  \subfigure[Running time (32 processors) vs. prefix size on a {\bf rMat graph} ($n = 2^{24}, m = 5 \times 10^7$) in log-log scale]{
    \hspace{-.015\linewidth}%
    \includegraphics[width=0.35\linewidth]{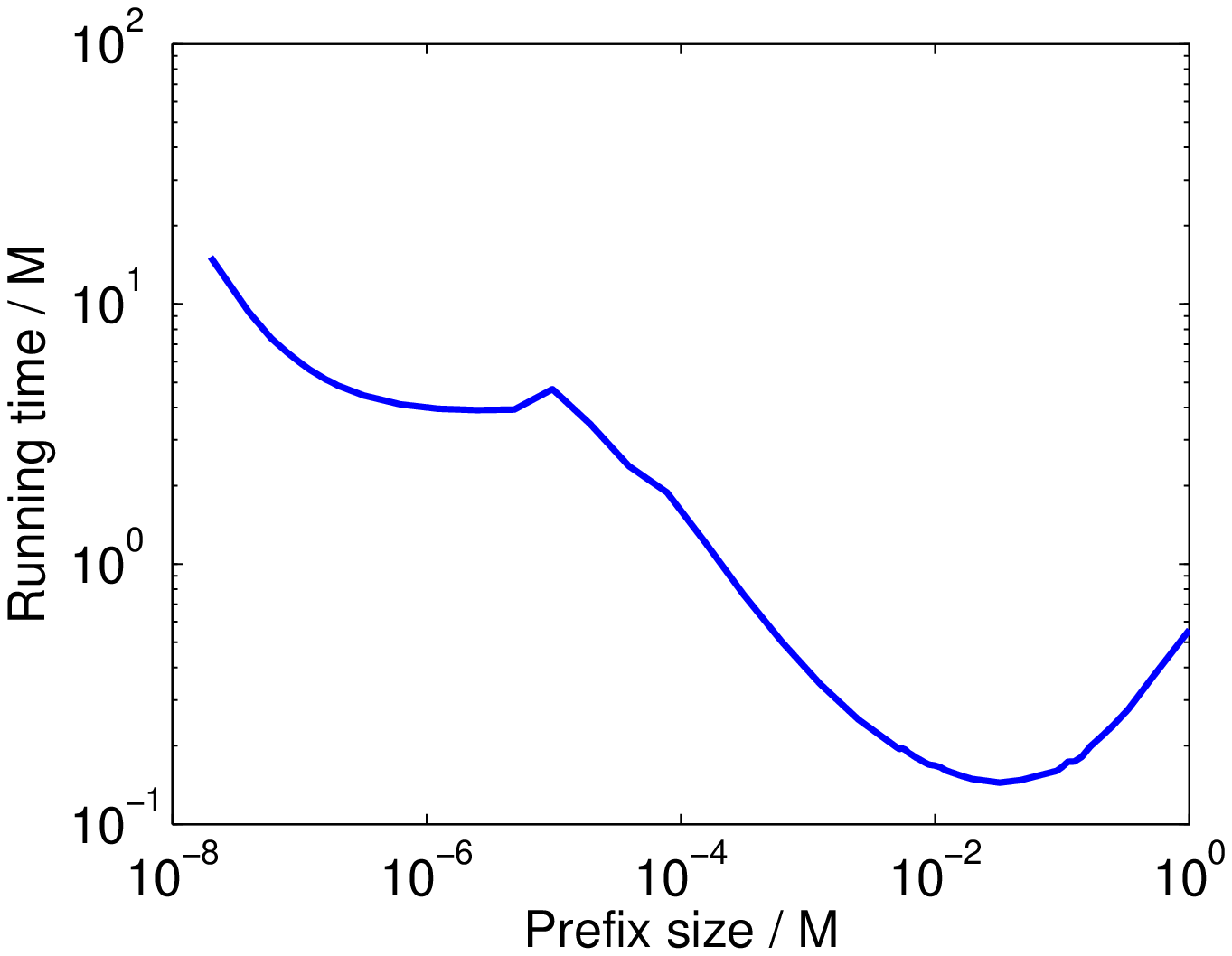}
    \hspace{-.03\linewidth}
    \label{fig:mmTimeRMat}
  }
  \caption{Plots showing the tradeoff between various properties and the prefix size in maximal matching.}
\end{figure*}

\begin{figure*}[h]
  \centering
  \subfigure[Running time vs. number of threads on a {\bf random graph} 
    ($n = 10^7, m = 5 \times 10^7$) in log-log scale]{
    \hspace{-.015\linewidth}
    \includegraphics[width=0.5\linewidth]{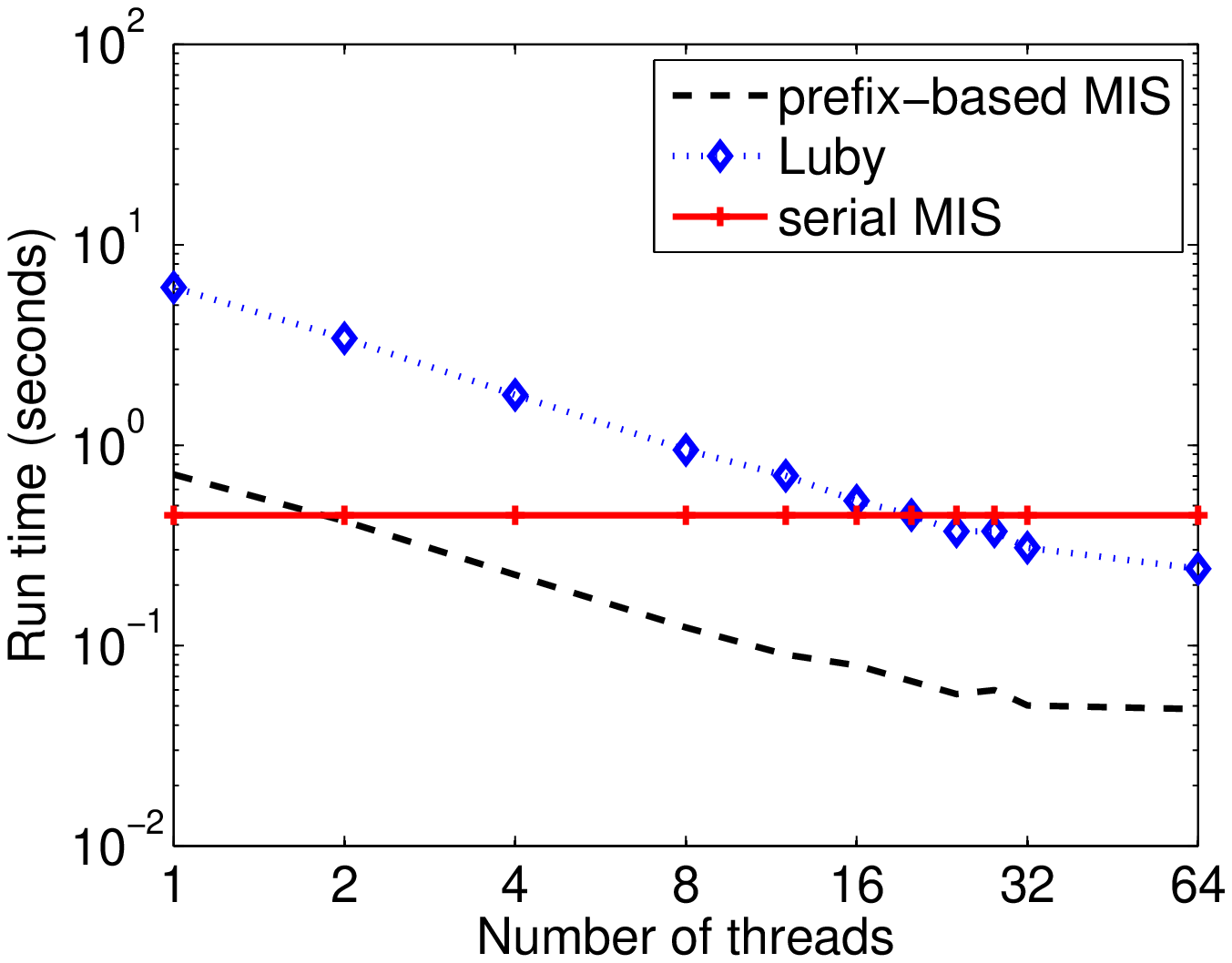}
    \hspace{-.015\linewidth}
    \label{fig:misPrefixVsProcsRand}
  } \hfill
  \subfigure[Running time vs. number of threads on a {\bf rMat graph} ($n = 2^{24}, m = 5 \times 10^7$) in log-log scale]{
    \hspace{-.015\linewidth}%
    \includegraphics[width=0.5\linewidth]{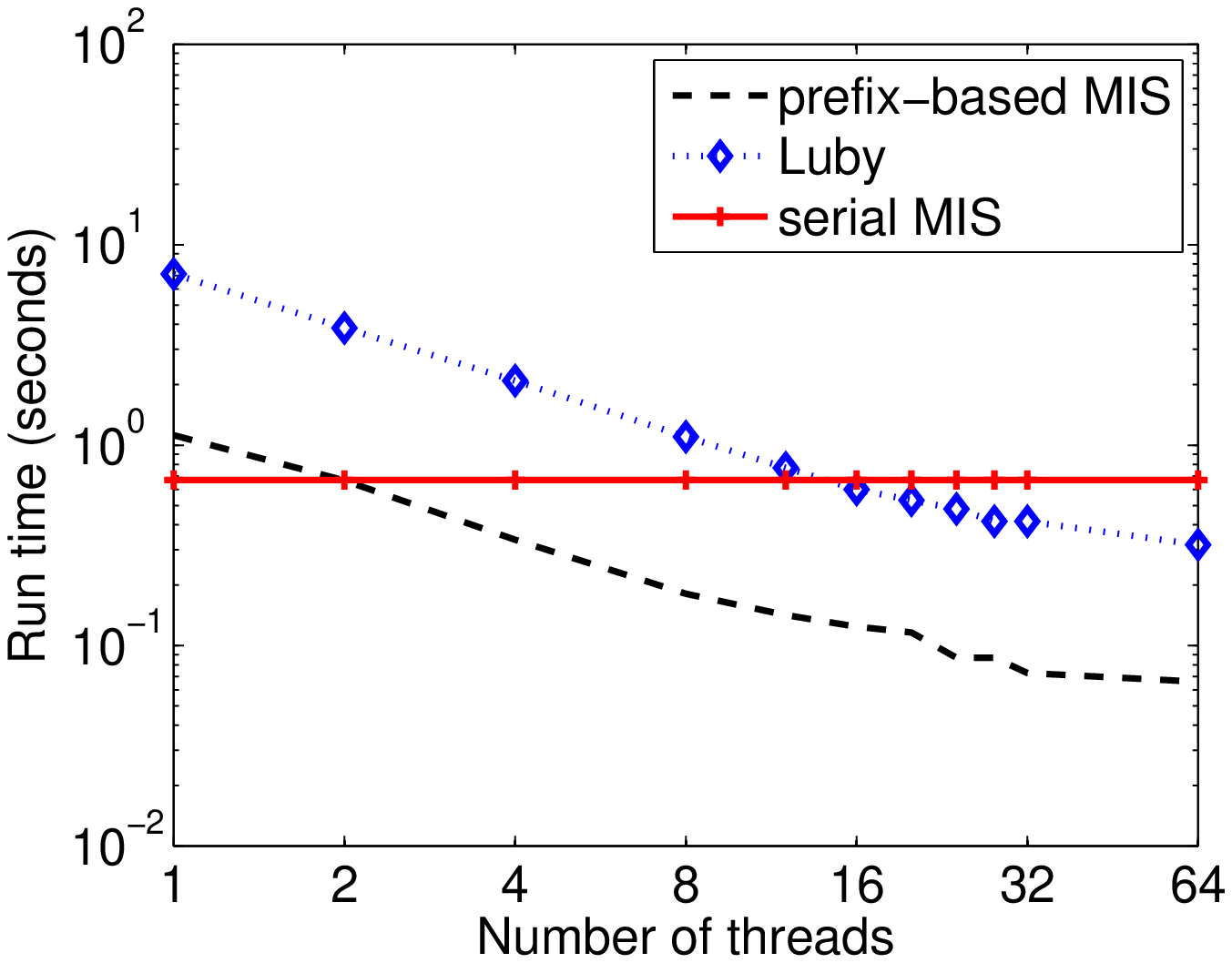}
    \hspace{-.015\linewidth}
    \label{fig:misPrefixVsProcsRMat}
  } \hfill
  \caption{Plots showing the running time vs. number of threads for the different MIS algorithms on a 32-core machine (with hyper-threading).}
\end{figure*}

\begin{figure*}[h]
  \centering
  \subfigure[Running time vs. number of threads on a {\bf random graph} 
    ($n = 10^7, m = 5 \times 10^7$) in log-log scale]{
    \hspace{-.015\linewidth}
    \includegraphics[width=0.5\linewidth]{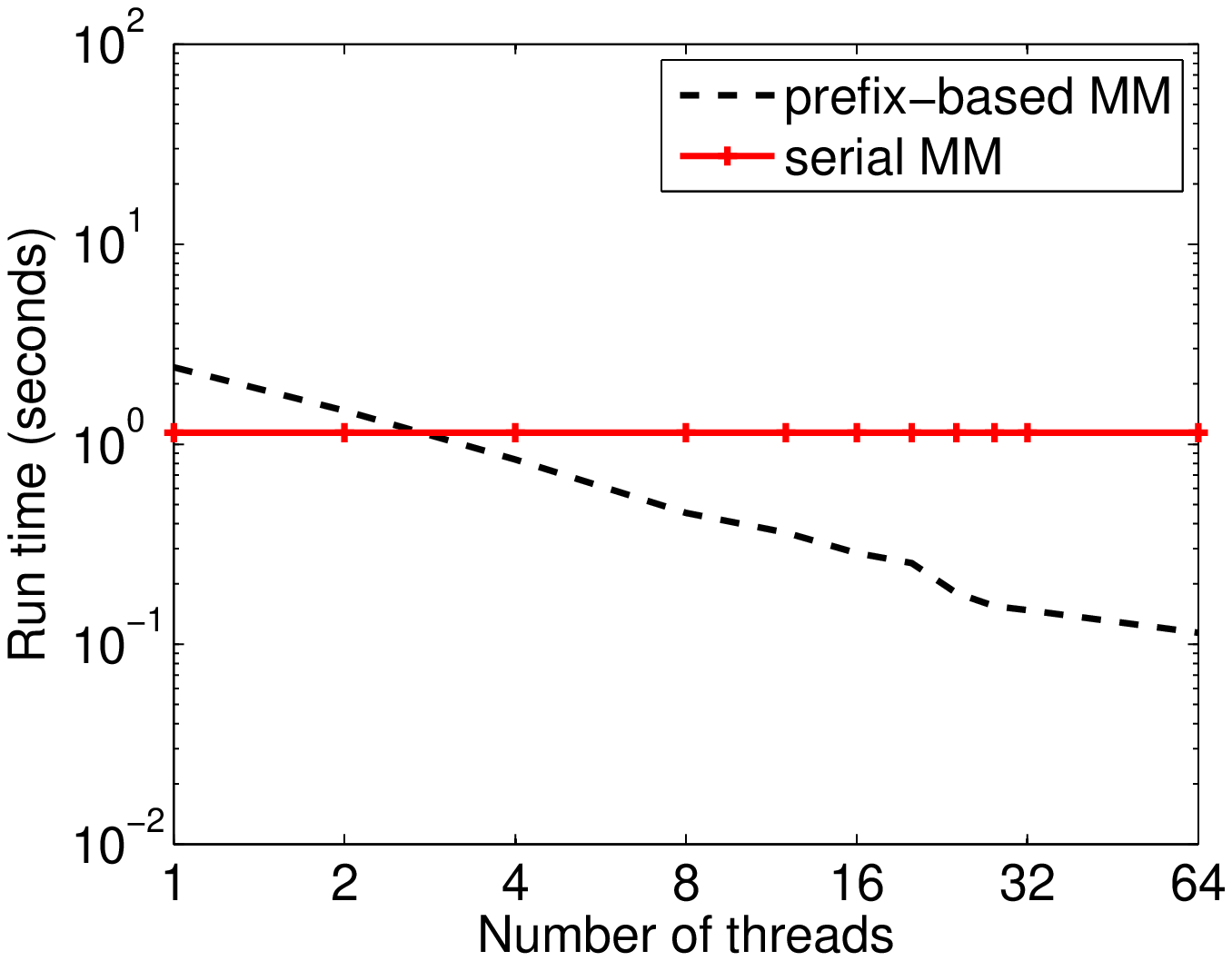}
    \hspace{-.015\linewidth}
    \label{fig:mmPrefixVsProcsRand}
  } \hfill
  \subfigure[Running time vs. number of threads on a {\bf rMat graph} ($n = 2^{24}, m = 5 \times 10^7$) in log-log scale]{
    \hspace{-.015\linewidth}%
    \includegraphics[width=0.5\linewidth]{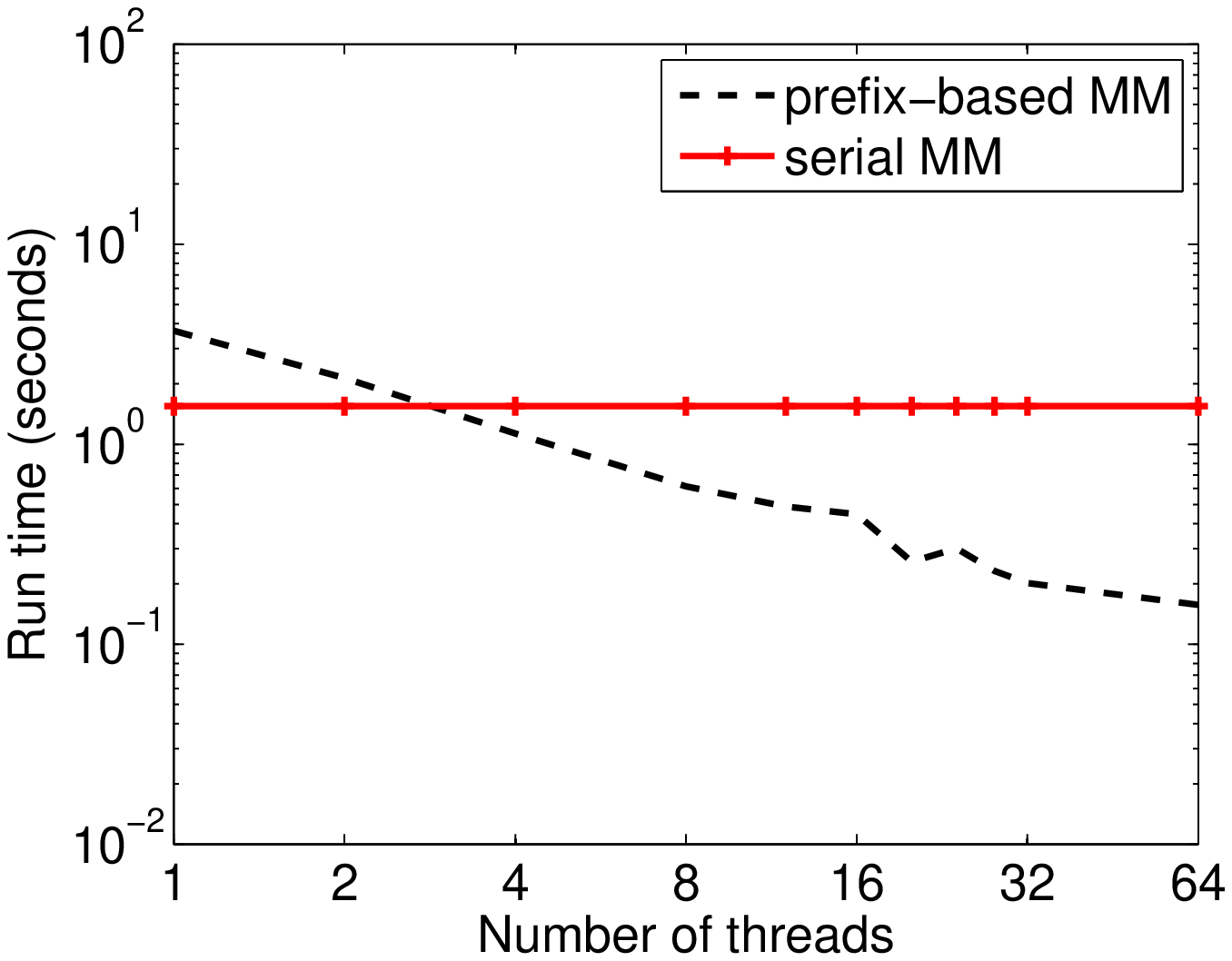}
    \hspace{-.015\linewidth}
    \label{fig:mmPrefixVsProcsRMat}
  } \hfill
  \caption{Plots showing the running time vs. number of threads for the different MM algorithms on a 32-core machine (with hyper-threading).}
\end{figure*}
%



\bibliographystyle{plain}
\bibliography{ref}


%
%
%
%

\end{document}